\documentclass{article}

\usepackage[T1]{fontenc}
\usepackage[utf8]{inputenc} 
\usepackage{lmodern}
\usepackage{amsmath,amssymb,amsfonts,amsthm,mathtools}

\usepackage[british]{babel}

\usepackage{graphicx}
\usepackage{subcaption} 
\usepackage{tikz}
\usetikzlibrary{arrows.meta,positioning,calc,patterns,decorations.pathreplacing}
\usepackage{tikz-cd}
\usepackage{pgfplots}
\usepgfplotslibrary{fillbetween}
\pgfplotsset{compat=1.18}
\usepackage{tkz-graph}

\usepackage{algorithm}
\usepackage{algpseudocode}

\usepackage{xcolor}
\usepackage{enumitem}
\usepackage{url}
\usepackage{fancyhdr}
\usepackage{indentfirst}
\usepackage{csquotes}
\usepackage{dsfont}     

\usepackage[numbers,sort&compress]{natbib}  
\providecommand{\textcite}{\citet}

\providecommand{\autocite}{\citep}

\providecommand{\citeauthor}{\citet}

\usepackage[colorlinks=true,citecolor=blue,urlcolor=blue]{hyperref}
\usepackage{bookmark} 
\usepackage{cleveref}

\newcommand{\E}{\mathbb{E}}

\theoremstyle{plain}
\newtheorem{theorem}{Theorem}[subsection]

\newtheorem{proposition}{Proposition}[subsection]
\newtheorem{conjecture}{Conjecture}[subsection]

\theoremstyle{definition}
\newtheorem{definition}{Definition}[subsection]

\theoremstyle{remark}
\newtheorem{remark}{Remark}[subsection]

\usepackage{mathrsfs}

\renewcommand{\geq}{\geqslant}
\renewcommand{\leq}{\leqslant}
\renewcommand{\epsilon}{\varepsilon}

\newcommand{\SupervisorName}{}         
\newcommand{\supervisor}[1]{\renewcommand{\SupervisorName}{#1}}

\begin{document}

\begin{titlepage}
  \begin{center}
    \vspace*{1cm}
    \Huge
    Convex Order and Arbitrage

    \vspace{1.5cm}
    \Large
    \textbf{Erica Zhang}

    \vspace{0.5cm}
    Supervisors:\\
    Professor Marcel Nutz,\\
    Professor Johannes Wiesel

    \vfill
    A thesis presented for the Honors degree of\\
    B.A.\ in Mathematics–Statistics

    \vspace{0.8cm}
    \includegraphics[width=0.4\textwidth]{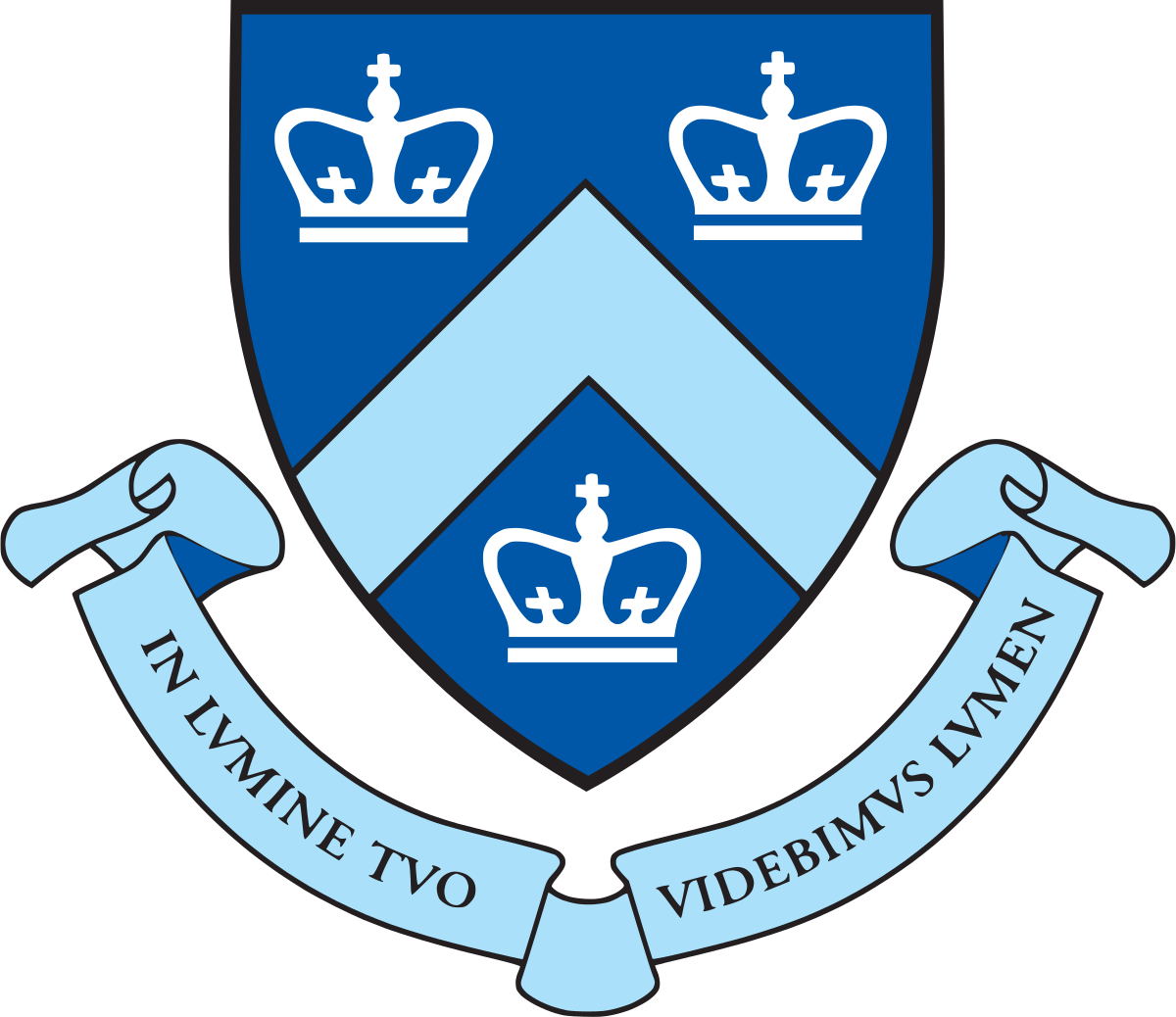}

    \Large
    Department of Mathematics\\
    Department of Statistics \\
    Columbia University\\
    March 31st, 2023
  \end{center}
\end{titlepage}

\section*{Acknowledgements}

Writing this thesis would not have been possible without the mentorship and support from my supervisors. I would like to thank Professor Marcel Nutz for kindly agreeing to supervise my thesis, despite having a heavily packed schedule. I sincerely appreciate his feedback and patience. I would like to also thank my supervisor Professor Johannes Wiesel for giving me an incredible amount of guidance and support throughout my time at Columbia. I would like to thank him for mentoring me with great patience and kindness, introducing me to the topic of optimal transport, and inspiring me to develop a keen interest in stochastic processes and mathematical finance. His immense knowledge and passion for his field of study have greatly encouraged me in many aspects of my academic research. My growth today would have been very difficult to achieve without him. For that, I am eternally grateful.

I would also like to express my gratitude for Professor Ioannis Karatzas, who kindly gave me guidance and enlightened me with key insights on functional portfolio theory during his free time, despite my not being enrolled in his class. I would like to thank him for the feedback he gave me on Section 6 of the thesis as well as key ideas that helped formulate Conjecture 6.3.1. I would like to thank him for sharing his passion in functional portfolio theory with me as well as introducing me to the brand-new field of stochastic portfolio theory.

\clearpage

{
  \hypersetup{linkcolor=black}
  \tableofcontents
}

\title{Convex Order and Arbitrage}
\author{Erica Zhang}
\supervisor{Prof. A.B. Supervisor}
\date{March 2023}

\maketitle

\begin{abstract}
Wiesel and Zhang [2023] \cite{wiesel2023optimal} established that two probability measures $\mu,\nu$ on $\mathbb{R}^d$ with finite second moments are in convex order (i.e. $\mu \preceq_c \nu$) if and only if $W_2(\nu,\rho)^2-W_2(\mu,\rho)^2 \leq \int |y|^2\nu(dy) - \int |x|^2\mu(dx).$ Let us call a measure $\rho$ maximizing $W_2(\nu,\rho)^2-W_2(\mu,\rho)^2$ the optimal $\rho$. This paper summarizes key findings by Wiesel and Zhang, develops new algorithms enhancing the search of optimal $\rho$, and builds on the paper through constructing a model-independent arbitrage strategy and developing associated numerical methods via the convex function recovered from the optimal $\rho$ through Brenier's theorem. In addition to examining the link between convex order and arbitrage through the lens of optimal transport, the paper also gives a brief survey of functionally generated portfolio in stochastic portfolio theory and offers a conjecture of the link between convex order and arbitrage between two functionally generated portfolios. 
\end{abstract}

\section{Introduction and Main Results}
Fix two probability measures $\mu,\nu \in \mathcal{P}(\mathbb{R}^d)$ with 
$$\int |x| \mu(dx) < \infty, \ \int |y| \nu(dy) < \infty.$$
Then $\mu$ and $\nu$ are defined to be in convex order (denoted by $\mu \preceq_c \nu$) iff $$\int f d\mu \leq \int f d\nu \ \ \text{for all convex functions }f: \mathbb{R}^d \rightarrow \mathbb{R}.$$ The above integrals take values in $(-\infty,\infty]$ as any convex function is bounded from below by an affine function. Convex order is a form of stochastic ordering, and it plays an essential role in mathematical finance. One of the most important theories in this domain is Strassen's theorem [1965] \cite{strassen1965existence}, which states that $\mu \preceq_c \nu$ if and only if $\mathcal{M}(\mu,\nu)$ - the set of martingale laws on $\mathbb{R}^d \times \mathbb{R}^d$ with marginals $\mu$ and $\nu$ - is non-empty. The well-established first fundamental theorem of asset pricing (FTAP) states that a market model has no arbitrage if and only if there exists a martingale deflator. If we can somehow make sense of $\mu$ and $\nu$ as certain financial instruments in the market, we can combine Strassen's theorem with the first FTAP and develop a test on arbitrage through convex order. This is our original motivation for the paper. A natural link of probability measures $\mu$ and $\nu$ with assets in the market is through Breeden-Litzenberger formula [1979] \cite{breeden1978prices}, where $\mu$ and $\nu$ represent laws of stock prices at maturities. The details of the formula are discussed in Section 1.2. 

There are a number of applications of convex order in the financial market as well as in financial modeling. To name a few, Guyon [2019] \cite{guyon2019inversion} studied the inversion of convex ordering in the VIX market and established that a continuous model on the S\&P 500 index (SPX) jointly calibrating to a full surface of SPX implied volatilities and to the VIX smiles exists if and only if for each time $t$, the local variance is smaller than the instantaneous variance in convex order. Carr et al. [2018] \cite{carr2018options} focused on convex order in risk-neutral densities for call options. In examining realized variance option and options on quadratic variation normalized to unit expectation, they found that if monotonicity in maturity for call options at a fixed strike is satisfied, the risk-neutral densities are increasing in the convex order. In addition, they also developed modeling strategies guaranteeing an increase in convex order for the normalized quadratic variation. Last but not least, the role of convex order in dependence modelling and risk aggregation is discussed extensively in paper such as Tchen [1980] \cite{tchen1980inequalities} and Bernard et al. [2016] \cite{bechmann2016variance}. 

In this paper, we offer a novel perspective by connecting convex order with the existence of arbitrage opportunities in the market and by developing trading strategies that leverage convex order in financial market. Specifically, the paper considers arbitrage in the trading of financial assets in a market and relative arbitrage between two functionally generated portfolios, with a focus on the former. Numerical implementation in this paper can be found in our \href{https://github.com/Ericavanee/Convex-Order-and-Arbitrage}{\color{blue}{Github repository}}. To make any of the convex order based arbitrage strategies practically realizable, it is critical to be able to identify convex order in general dimensions. To do this, we leverage the new optimal transport-based characerization of convex order method developed by Wiesel and Zhang, introduced in the section below. 

\subsection{Convex Order through the 2-Wasserstein Distance}
In paper \href{https://papers.ssrn.com/sol3/papers.cfm?abstract_id=4121931}{``An optimal transport based characterization of convex order"}, we, Wiesel and Zhang [2023] \cite{wiesel2023optimal} developed a new method of characterizing convex order in general dimensions using tools from optimal transport. Specifically, two probability measures $\mu,\nu$ with finite second moments are in convex order (denoted $\mu \preceq_c \nu$) if and only if 
\begin{align}
W_2(\nu,\rho)^2-W_2(\mu,\rho)^2 \leq \int |y|^2\nu(dy) - \int |x|^2\mu(dx)
\end{align}

This result is obtained through combining classical optimal transport duality with Brenier's theorem, see Brenier [1991] \cite{brenier1991polar}. Optimal transport goes back to the seminal works of Monge [1981] \cite{monge1781memoire} and Kantorovich [1958] \cite{kantorovich1958translocation}. It is concerned with the problem of transporting probability distributions in a cost-optimal way. We refer to Rachev and Rüschendorf [1998] \cite{rachev1998mass} and Villani [2003, 2008] \cite{villani2008optimal} for an overview. In this paper, Wiesel and Zhang attempt to develop a method to characterize convex order in general dimensions beyond the real line. To the best of our knowledge, this link between convex order and optimal transport can first be found in Carlier [2008] \cite{carlier2008remarks}. Contrary to Carlier's proof, however, we make no assumptions on the probability measures $\mu,\nu$ except for that they have finite second moments. In particular, we prove the theorem in a general case where there is no need to assume that the measures are absolutely continuous wrt. the Lebesgue measure (see proof of Theorem 1.1 in Wiesel and Zhang \cite{wiesel2023optimal}). 

To develop model-independent strategy based on this optimal transport-based characterization, we introduce the Breeden-Litzenberger formula. 

\subsection{Breeden-Litzenberger formula}
The Breeden-Litzenberger (B-L) formula is a fundamental tool in option pricing theory that has been widely used in finance and economics since its introduction in 1978, see Breeden and Litzenberger [1978] \cite{breeden1978prices}. The formula builds on the concept of state-contingent financial claims, which are financial instruments whose payoff depends on the occurrence or non-occurrence of certain events. Options are a common example of state-contingent claims, as their payoffs depend on the price of the underlying asset at expiration. The B-L formula relates the prices of options on the same underlying asset with different strike prices and maturities to the prices of the underlying asset and other financial instruments, such as bonds. The set-up is the following.

Consider a market with a risk-free asset with interest rate $r =0$ and $d$ risky assets. Assume that the risky asset $S = (S^1,...,S^d)$ over a time horizon $T$ is a Markov process. Now consider a European call option portfolio with options at the same maturity $T$ and traded at all strikes. Then the B-L formula enables us to recover the joint probability distribution of the underlying risky assets by observing the call option prices in the market's European call option portfolios. 

We give a brief overview of the derivation of the B-L results. Fix a maturity $T$. Consider a martingale measure $\mu$ and denote the law of $S$ under $\mu$ by $\mathbb{P}$. Then the price of a general European contingent call option of maturity $T$ and strike $K$, denoted $C(T,K)$, is given by the following:
\begin{align*}
   C(T,K) &= \mathbb{E}_{\mu}[(S_T-K)^+] \\ &= \int_K^\infty (x-K)\mu(dx)
\end{align*}
Thus,
\begin{align*}
    \frac{\partial}{\partial K}C(T,K) &= \int_{K}^\infty -\mu(dx) \\ &=-\mu([S_T \geq K])
\end{align*}
and 
\begin{align*}
    \frac{\partial^2}{\partial K^2}C(T,K) &= \mu(dK)
\end{align*}
implying that observing all call prices yields the marginal density of $S_T$. The above results is exactly the B-L formula.

Indeed, one of the key insights of the B-L formula is that the risk-neutral probability of the underlying asset's price changing by a certain amount can be inferred from the prices of options on the same asset. This allows one to price options using the Black-Scholes formula or other option pricing models that assume a constant volatility, even if the actual volatility of the underlying asset is unknown or time-varying. However, some pratical issues remain. For example, interpolating observed call prices with a smooth function is very difficult, and the interpolation scheme might not be arbitrage-free. Additionally, the formula relies on the assumption of no-arbitrage, which may not hold in certain market conditions or during periods of market stress. Several papers have extended or modified the B-L formula to address some of its limitations. For instance, Rubinstein [1994] \cite{rubinstein_1994} proposed a modification to the B-L formula that allows for non-log-normal underlying asset distributions. Dumas, Fleming, and Whaley [1998] \cite{dumas_fleming_whaley_1998} proposed a modified B-L formula for valuing interest rate options that accounts for term structure effects. Wiggins [1987] \cite{wiggins_1987} proposed a method for estimating the implied risk-neutral probability density function using the B-L formula.

In our paper, we restrict ourselves to the classical setting of the B-L formula and leave the discussion of relaxed assumptions to another paper. We then introduce the final element for the construction of our model-free trading strategy. 

\subsection{Model-free version of arbitrage strategy}
We consider in this paper the so-called \textit{model-independent approach} to the study of arbitrage and trading strategy, i.e. the consideration of the set of all models contingent with the prices observed in the market. The search of robust methods to mathematical finance problems such as asset-pricing and super-replications has been gaining traction especially recently. Hobson [1998] \cite{hobson_1998} pioneered the study of \textit{model-independent arbitrage} with respect to \textit{semi-trading strategies} and robust hedging. Since then, numerous literature considering the necessary and sufficient conditions for the existence of calibrated arbitrage-free models under different settings given an observation of prices of finitely (or infinitely) European call options is developed, see Cousot [2007] \cite{cousot_2004}, Buehler [2006] \cite{buehler_2006}, and Carr and Madan [2005] \cite{carr_madan_2005}. The model-free version of the Fundamental Theorem of Asset Pricing is studied by Riedel [2011] \cite{riedel_2011} and Acciaio et al. [2018] assumes the prices of finitely many options and considers the robust version of the FTAP with respect to a weak notion of arbitrage under the assumption of a compact state-space in a one-period setting. Moreover,  given a path-dependent derivative $\Phi$, Cox and Obloj [2011] \cite{cox_obloj_2011a} and others made a series of important advancements in the development of a robust version of the \textit{Super-Replication Theorem}, where they determined explicitly and demonstrated the equivalence in values of the values of upper martingale price $p^M(\Phi)$, i.e. the supremum of the expected payoff over admissible martingale measures and the minimal endowment $p^R(\Phi)$ required for super-replication. Following their work, Nutz [2013] \cite{nutz_2013} extends the setup and studies optimal super-replication strategies where super-replication is understood not in a point-wise sense, but rather w.r.t a family of probability measures.

In addition to robust trading strategies developed in a more ``traditional", stochastic integral-based mathematical finance literature, stochastic portfolio theory, such as the idea of the \textit{functionally generated portfolio}, also enables the construction of a robust trading strategy that does not make any assumptions on the price dynamics of the underlying assets, see Fernholtz, Karatzas, Ruf [2018] \cite{fernholz2016volatility} for an example. Recently, there is also an emergence of literature on combining tools from machine learning with traditional mathematical models that describe the dynamics of underlying financial assets to develop robust statistical tools for financial markets (in option pricing or calibration) and arbitrage. Gierjatowicz et al. [2020] \cite{article}, for example, studied the performance of \textit{neural-SDE} where they combine feed-forward neural networks with a stochastic volatility model to construct robust hedging and pricing techniques. 

Up to Section 6, we restrict ourselves to the concept of \textit{model-independent arbitrage} via semi-static trading strategies, as introduced by Davis and Hobson [2007] \cite{davis_hobson_2007}. The resulting trading strategy for model-independent arbitrage is commonly known as the  \textit{calendar spread}. We refer to a more elaborate discussion in Section 5. 

\section{Setup and Notation}
For the majority of the paper (up to Section 6), we follow the same market setup as in Acciaio [2013] \cite{acciaio_et_al_2018}. We consider a finite, discrete time setting with some time horizon $T$ and a market consisting of $d$ Markovian risky assets with price process $(S_t)_{t = 0}^T = (S_t^1,...,S_t^d)_{t = 0}^T$ and a normalized risk free asset $B = (B_t)_{t=0}^T$.

For the purpose of discussing the optimal transport-based characterization of convex order, we consider probability measures $\mu,\nu \in \mathcal{P}(\mathbb{R}^d)$ with finite first moments. We define cost functionals $$C(\mu,\rho):= \sup_{\pi \in \Pi(\mu,\rho)}\int \langle x,y \rangle \pi(dx,dy), \ \ C(\nu,\rho):= \sup_{\pi \in \Pi(\nu,\rho)}\int \langle x,y \rangle \pi(dx,dy),$$
where $\langle \cdot, \cdot \rangle$ denotes the scalar product and $\Pi(\cdot,\cdot)$ denotes the set of couplings. As in the paper by Wiesel and Zhang, we define $$\mathcal{P}^1(\mathbb{R}^d):= \{\rho \in \mathcal{P}(\mathbb{R}^d): \text{supp}(\rho) \subseteq B_1(0)\},$$ where $B_1(0)$ denotes the Euclidean ball around $0$ with a radius of $1$. In addition, for two matrices $\textbf{M},\textbf{N}$, we denote matrix multiplication simply via $\textbf{M}\textbf{N}$. We use $\textbf{M}^T$ to denote the transpose of a matrix $\textbf{M}$. We use bold font of mathematical symbols to denote its vector form, and we denote the discretized samples from distributions $\mu,\nu, \rho$ as $\textbf{a}$, $\textbf{b}$, $\textbf{r}$ respectively. Denote the discretized matrix of joint probability distribution $\pi$ as $\gamma$.

\section{Optimizing \texorpdfstring{$\rho$}{rho}}
 
In this section, we discuss numerical algorithms we used in search of the optimal $\rho^* \in \mathcal{P}^1(\mathbb{R}^d)$ with domain $B_1(0)$ such that 
\begin{align*}
\rho^* = \inf_{\rho \in \mathcal{P}^1(\mathbb{R}^d)}(C(\nu,\rho)-C(\mu,\rho))
\end{align*}
Recall that we denote the discretized samples from distributions $\mu,\nu, \rho$ as $\textbf{a}$, $\textbf{b}$, $\textbf{r}$ respectively. If measure $\mu$, $\nu$ are finitely supported, then we can solve the problem numerically by solving the earth mover's distance (EMD) problem:
\begin{align*}
\min_{\textbf{r} \in B_1(0)} \{-\min_{\gamma_\mu} \{ M_\mu\gamma_\mu \} +  \min_{\gamma_\nu} \{M_\nu\gamma_\nu \}\} \hspace{1cm}
\end{align*}
where $M_\mu,M_\nu$ denote repectively the squared-Euclidean distance matrix between $\mu$ and $\rho$ and $\rho$ and $\nu$, such that 
$\gamma_\mu\mathbf{1} = \textbf{a}, \gamma_\mu^{\textit{T}}\mathbf{1} = \textbf{r}, \pi_\nu\mathbf{1} = \textbf{b}, \gamma_\nu^{\textit{T}}\mathbf{1} = \textbf{r}$ and $\gamma_\mu,\gamma_\nu \geq 0$, where $\textbf{1}$ denote a column matrix of 1s.
\newline
\newline 
It is clear from the definition that the optimization of measure $\rho$ hinges on the numerical exploration of the convex set $\mathcal{P}^1({\mathbb{R}^d})$. We (Wiesel and Zhang) developed two Bayesian optimization-based methods for this. Implementation can be found in the \href{https://github.com/johanneswiesel/Convex-Order}{Github repository}. In the implementation, we use the \href{https://pythonot.github.io}{POT package} to compute optimal transport distances and \href{http://hyperopt.github.io/hyperopt/}{HyperOpt} for Bayesian optimization. 

\subsection{Indirect Dirichlet Sampling}
We consider finitely supported $d-$dimensional measures $\rho$ on an $p$-evenly partitioned grid $G$ on $B_1(0)$, which are dense in $\mathcal{P}^1(\mathbb{R}^d)$ in the Wasserstein topology. The measure $\rho$ is sampled from the Dirichlet distribution on the space $\mathbb{R}^{g-1}$, $g \in \mathbb{N}$, with density 
\begin{align*}
f(x_1,...,x_{g}; \alpha_1,...,\alpha_{g}) = \frac{1}{B(\pmb{\alpha})}\Pi_{i=1}^{g} x_i^{\alpha_i-1}
\end{align*}
for $x_1,...,x_{g} \in [0,1]$ satisfying $\sum_{i=1}^{g} x_i = 1$. Here $\alpha_1,...,\alpha_{g} >0$, $\pmb{\alpha} := (\alpha_1,...,\alpha_{g})$ and $B(\pmb{\alpha})$ denotes the Beta function. Fixing $g$ grid points $\{k_1, . . . , k_g\}$ in $B_1(0)$, we are able to indirectly sample $\rho$ by considering any realization of a Dirichlet random variable $(X_1,...,X_{g})$ as a probability distribution assigning probability mass $X_i$ to the grid point $k_i$, $i \in \{1,...,g\}$, on which $\rho$ is supported upon. We name this the ``Indirect Dirichlet" method. Let us define $$\rho(\pmb{\alpha}) = \sum_{i=1}^{g}X_i\delta_{k_i}.$$ To sample the optimal probability measure $\rho^*(\pmb{\alpha})$ parametrized by $\pmb{\alpha}$, we exploit Bayesian hyperparameter optimization of parameter $\pmb{\alpha}$ via a Tree-structured Parzen Estimator Approach (TPE) (Bergstra, Bardenet et al.) to optimize $\pmb{\alpha}$. This leads to the following algorithm:

\begin{algorithm}[H]
\caption{Basic algorithm for Indirect Dirichlet method}

\textbf{Input:} samples (or histograms) $\textbf{a}$, $\textbf{b}$ drawn respectively from probability measures $\mu,\nu$, maximal number of Bayesian optimization evaluations $N$, optimizing parameter search space lower bound $l$ and upper bound $u$, number of discretized grid points $g$. 
\newline
\textbf{Initialization}: $\pmb{\alpha} := (\alpha_1,...,\alpha_{g})$. 
\begin{algorithmic}
\State Generate a grid $B_1(0)$ of $g$ equidistant points.
\For{Trial 1: $N$}
    \State Sample a Dirichlet random variable $\rho(\pmb{\alpha})$ parametrized by $\pmb{\alpha}$ supported on $M$. Use TPE to update $\pmb{\alpha}:= (\alpha_1,...,\alpha_{g})$, where $$\pmb{\alpha}^* = \arg\min_{\pmb{\alpha} \in [l,u]^{g}}(C(\rho(\pmb{\alpha}),\nu)-C(\rho(\pmb{\alpha}),\mu))$$
\EndFor
\State \Return $C(\rho(\pmb{\alpha}^*),\nu)-C(\rho(\pmb{\alpha}^*),\mu)$, $\pmb{\alpha}^*$.
\end{algorithmic}
\end{algorithm}
A major computational challenge in Algorithm 1 is the efficient evaluation of $C(\rho(\pmb{\alpha}),\nu)$ and $C(\rho(\pmb{\alpha}),\mu)$. We offer two variants:
\begin{itemize}
\item \textit{Indirect Dirichlet method with histograms:} If we have access to finitely supported approximations $\textbf{a}$ and $\textbf{b}$ of $\mu$ and $\nu$ respectively and the measure $\rho$ is supported on $G$, then we can solve the linear programs $C(\textbf{a},\rho)$ and $C(\textbf{b},\rho)$ as is standard in optimal transport theory. Due to its relative lower complexity, this method enjoys a higher computational efficiency. Note that the histogram method is generally available to 1D transportation problems. We refer to \texttt{histograms.py} for a sampling from commonly seen distributions. 
\item \textit{Indirect Dirichlet method with samples:} we draw from a number of samples from $\mu$ and $\nu$ respectively and denote these empirical samples as $\textbf{a}, \textbf{b}$ respectively. As before, we assume that the probability measure $\rho$ is finitely supported on a grid $G$. We then solve the linear programs $C(\textbf{a},\rho)$ and $C(\textbf{b},\rho)$. Because the transportation costs need to be evaluated at each sample point, this method has a higher complexity than the histograms method above, although it is useful when a histogram is computationally expensive to obtain. A more detailed discussion of analysis of runtime is discussed in section 2.4.
\end{itemize}

All in all, the efficient evaluation of global optimizer $\pmb{\alpha} = (\alpha,...,\alpha_{g})$ for problem $\inf (C(\rho(\pmb{\alpha}),\nu)-C(\rho(\pmb{\alpha}),\mu))$ depends on the effectiveness of the indirect Dirichlet method and the Bayesian optimization procedure to explore the convex space $\mathcal{P}^1(\mathbb{R}^d)$. For this, we offer another way to model $\rho$ as well as propose another potential optimizing procedure, as described in the sections below. 

\subsection{Direct Randomized Dirichlet Sampling}
An alternative to Algorithm 1 is to directly draw samples from a distribution $\rho \in \mathcal{P}^1(\mathbb{R}^d)$ based on a Dirichlet distribution. In this case, $\rho$ is determined by the Dirichlet distribution and then randomly sampled.  We refer to this as the ``Direct randomized Dirichlet method". 

\begin{algorithm}[H]
\caption{Basic algorithm for Direct randomized Dirichlet method}\label{alg:cap2}

\textbf{Input:} samples drawn from probability measures $\mu,\nu$, maximal number of Bayesian optimization evaluations $N$, optimizing parameter search space lower bound $l$ and upper bound $u$, dimension $d$ of measure $\rho \in \mathcal{P}^1(\mathbb{R}^d)$, target size $t$ of samples drawn from $\rho$. 
\newline
\textbf{Initialization}: $\pmb{\alpha} := (\alpha_1,...,\alpha_{d+1})$. 
\begin{algorithmic}
\For{Trial 1: $N$}
    \State Initiate an array $A$ to contain $t$ samples drawn from Dirichlet random variable $D(\pmb{\alpha})$ with randomized signs: $A = [\textbf{D}_1,..,\textbf{D}_t]$, where each sample $\textbf{D}_i$ has dimension $d+1$. Then flatten array $A$.   \While{$\sqrt{\sum_{i=1}^t(\textbf{D}_i)^2} > 1$}
        \State Select without replacement $t$ $d-$dimensional vectors $\textbf{D}_1,...,\textbf{D}_t$.
    \EndWhile      
    \State Then $$\rho = \frac{1}{t}\sum_{i=1}^t\delta_{\textbf{D}_i}.$$ Use TPE to update $\pmb{\alpha} := (\alpha_1,...,\alpha_{d+1})$, where $$\pmb{\alpha}^* = \arg\min_{\pmb{\alpha} \in [l,u]^{d+1}}(C(\rho(\pmb{\alpha}),\nu)-C(\rho(\pmb{\alpha}),\mu))$$
\EndFor
\State \Return $C(\rho(\pmb{\alpha}^*),\nu)-C(\rho(\pmb{\alpha}^*),\mu))$, $\pmb{\alpha}^*$, $\rho(\pmb{\alpha}^*)$.
\end{algorithmic}
\end{algorithm}

\subsection{Optimization with \texttt{cvxpy}}
An alternative approach to optimize $\rho$ would be to apply convex optimization algorithms in \href{https://www.cvxpy.org/}{\texttt{cvxpy}}. It is well known that the Wasserstein metric can be computed as a linear program as formulated by the Kantorovich-Rubinstein duality. For the purpose of the algorithm, we consider like before discrete distributions $\mu,\nu$ on a grid $G$ on the Euclidean ball $B_1(0)$. Using the squared $L^2$ norm, we still use the squared-Euclidean distance for the distance matrix $M$. Then the p-Wasserstein distance between $\mu,\nu$ can be computed numerically: $$W_p(\mu,\nu) = (\min_{\gamma}M\gamma)^{\frac{1}{p}} = (\min_{\gamma} \sum_{i,j}M_{i,j}\gamma_{i,j})^{\frac{1}{p}},$$ such that $\gamma\textbf{1} = \textbf{a}$, $\gamma^T\textbf{1} = \textbf{b}$, and $\gamma \geq 0$. We offer here the \texttt{cvxpy} algorithm for computing the Wasserstein metric between probability measures $\mu,\nu \in \mathcal{P}^1(\mathbb{R}^d)$ in general dimensions using samples. We assume uniform sample weights. And that sample size of $\textbf{a}$ drawn from $\mu$ is $n$, and that of $\textbf{b}$ drawn from $\nu$ is $m$. Also, let $\textbf{m},\textbf{n}$ denote column vector of $m,n$ respectively.

\begin{algorithm}[H]
\caption{Basic cvxpy algorithm for computing the Wasserstein metric}
\textbf{Input:} samples $\textbf{a}$, $\textbf{b}$ drawn from probability measure $\mu, \nu$. 
\newline
\textbf{Initialization}: optimal coupling between $\mu$ and $\nu$ is initialized as a non-negative \texttt{cvxpy} variable $P$ of shape $(n,m)$. 
\begin{algorithmic}
\State Calculate distance matrix $M$. Apply \texttt{cvxpy} solver on problem
$$\min_{P} \text{trace}(M^TP)$$
under constraints
$$P\textbf{1} = 1/\textbf{n}, \ \ P^T\textbf{1} = 1/\textbf{m}.$$
\State \Return $\min_{P} \text{trace}(M^TP)$, $P$

\end{algorithmic}
\end{algorithm}

For the purpose of computing our loss function, we offer the general procedures of using cvxpy to solve our $2-$fold optimization problem by approximating the Wasserstein metric as a linear program computed as in the algorithm above. To see how this works, consider first the below problem:
\begin{align}
\inf_{\mu \in \mathcal{P}^1(\mathbb{R}^d)}C(\mu,\rho)+C(\rho,\nu),
\end{align}
which is a variation of the general Wasserstein Barycenter problem:
$$\arg\min_\mu \sum_{i = 1}^k \lambda_i\mathcal{W}(\mu,\nu_i),$$
where $\mu,\nu_1,...,\nu_k$ are probability distributions supported on $\mathbb{R}^d$, vector $\lambda \in \mathbb{R}^k$ of non-negative weights summing to 1, and $\mathcal{W}(\cdot)$ denotes the squared 2-Wasserstein distance. For the implementation of the algorithm, we have the below formulation of the problem, where we have chosen $L^2$-norm as the distance metric:
$$\min_{\gamma^k \in \Pi(\mu,\nu^k)} \sum_{k = 1}^{2} \sum_{i,j}\|X_i-Y_j\|_2^2\gamma_{i,j}^k$$ such that $\sum_{j =1}\gamma_{ij}^k = \textbf{a}_i \ \forall_{k,i}$, $\sum_{i = 1}\gamma^k_{ij} = \textbf{b}^k_j \ \forall_{k,j}$, and $\gamma^k_{ij} \geq 0 \ \forall_{k,i,j}$.

We now give a basic working algorithm to the problem (2) in 1D with distributions $\mu,\nu$ inputted as histograms over a fixed grid in a \texttt{cvxpy} framework. The output is the histogram $\textbf{r}$ of the Wasserstein Barycenter $\rho$ over the same grid. 
\begin{algorithm}[H]
\caption{Basic cvxpy algorithm for Wasserstein Barycenter in 1D}
\textbf{Input:} histogram of probability measure $\mu$ (denoted $\textbf{a}$), histogram of target probability measure $\nu$ (denoted $\textbf{b}$) with a fixed grid $g$ of size $n$.
\newline
\textbf{Initialization}: probability mass function of $\rho$ over grid $g$ initialized as a non-negative cvxpy variable of length $n$; optimal transport plan $\pi^\mu$ between $\mu,\rho$ and $\pi^\nu$ between $\rho,\nu$ initialized as non-negative \texttt{cvxpy} variable of shape $(n,n)$, and non-negative \texttt{cvxpy} variable $t^\mu,t^\nu$.

\begin{algorithmic}
\State Calculate fixed distance matrix  \textbf{M} of grid $G$.
\State Apply \texttt{cvxpy} solver on problem
$$\min_{\textbf{r}, \gamma^\mu, \gamma^\nu, t^\mu, t^\nu} t^\mu+t^\nu$$
under constraints
$$t^\mu >= \textbf{1}^TM\gamma^\mu\textbf{1}, \ \ t^\nu >= \textbf{1}^TM\gamma^\nu\textbf{1};$$

$$(\gamma^\mu\textbf{1})^T = \textbf{r}, \ \  (\gamma^\nu\textbf{1})^T = \textbf{r};$$

$$\gamma^\mu\textbf{1} = \textbf{a}, \ \ \gamma^\nu\textbf{1} = \textbf{b}.$$
\Return $\textbf{r}$
\end{algorithmic}
\end{algorithm}

We emphasize that the above algorithm only works for 1D, where inputs are histograms. This formulation of the problem can be effectively computed in \texttt{cvxpy} because the grid, and as a result the cost matrix computation, is fixed and does not vary with a different choice of $\rho$. To relax 1D histogram into general samples, one idea is to assign samples to the nearest grid points. We will leave this exploration to later work. 

Since the Wasserstein Barycenter problem is a convex problem, it can be solved in \texttt{cvxpy}. However, our targeted optimization problem $$\inf_{\rho \in \mathcal{P}^1{\mathbb{R}^d}}(C(\rho,\nu)-C(\rho,\mu))$$
is non-convex. Similar procedures as in Algorithm 3 applies when writing algorithm for the optimization problem. Assume that $\mu,\nu$ are $d-$dimensional probability distributions. We give below a general algorithm of how this problem would be implemented in \texttt{cvxpy}, were it to be convex. As before, we let $\textbf{m},\textbf{n}$, \textbf{t} denote column vector of natural number $m,n,t$ respectively.
\begin{algorithm}[H]
\caption{Tentative cvxpy algorithm for $\min_{\rho \in \mathcal{P}^1{\mathbb{R}^d}}(C(\rho,\nu)-C(\rho,\mu))$}
\textbf{Input:} $m$ samples drawn from probability measure $\mu$ (denoted $\textbf{a}$) $n$ samples drawn from probability measure $\nu$ (denoted $\textbf{b}$), sample size $t$ of measure $\rho$ (its samples denoted \textbf{r}). 
\newline
\textbf{Initialization}: samples \textbf{r} drawn from target distribution $\rho$ of shape ($t$,$d$); non-negative optimal coupling $P_a$ from $\mu$ and $\rho$ of shape ($m$,$t$) and non-negative optimal coupling $P_b$ from $\rho$ to $\nu$ of shape ($t$,$n$).
\begin{algorithmic}
\State Compute and express in \texttt{cvxpy} variable the squared Euclidean distance matrix $M_a$ from $\mu$ to $\rho$ and distance matrix $M_b$ from $\rho$ to $\nu$:
$$M_a = \|\textbf{a}-\textbf{r}\|_2^2, \ \ M_b = \|\textbf{b}-\textbf{r}\|_2^2$$
\State Apply cvxpy solver on problem
$$\min_{P_a,P_b} \text{trace}(M_a^TP_a)+\text{trace}(M_b^TP_b)$$
under constraints
$$P_a\textbf{1} = 1/\textbf{m}, \ P_a^T\textbf{1} = 1/\textbf{t}, \ P_b\textbf{1} = 1/\textbf{t}, \ P_b^T\textbf{1} = 1/\textbf{n}.$$
\Return $\textbf{r}$, $P_a$, $P_b$

\end{algorithmic}
\end{algorithm}

It is noteworthy that this is not a running algorithm due to the non-convexity of our cost function. Two restrictions are present. One is that by formulating the problem in a two-fold optimization problem (unlike the previous 1D formulation of the barycenter problem), the cost matrix becomes convoluted and is not DPE compliant. This is a problem when one tries to write a generalized dimension Wasserstein Barycenter solver in cvxpy, as for Algorithm 4, a fixed distance metric is key. While this is solvable as in the proposal above, the non-convexity of the problem fundamentally prevents it from being implemented in \texttt{cvxpy}.  

\subsection{Complexity and Performance}
We divide the section into two parts: (i). a simple overview and analysis of algorithm run-time and complexity and (ii). a presentation of experiments and results. 

\subsubsection{Analysis of Complexity}
All three of our working algorithms are based on the Bayesian optimization solver by \texttt{HyperOpt}. We begin with a discussion of the mechanism underlying Bayesian optimization and an analysis of its complexity using the adapted tree-structured parzen estimator approach (TPE). For a more elaborate discussion, we refer to Bergstra et al.[2013] \cite{bergstra2013making} and Shahriari et al.[2016] \cite{shahriari2016taking}.

Bayesian optimization is a type of hyper-parameter optimization which tackles the problem of optimizing a loss function over a graph-structured configuration space (or observation history) $\mathcal{H} = (x_i, f(x_i))^{n}_{i=1}$. In HyperOpt, a configuration space serves as the search space of hyper-parameters and is defined by certain generative process that draws valid samples conforming to constraints. The structure of the Bayesian optimization algorithm follows the prototype of Sequential Model-Based Global Optimization (SMBO) algorithms, which is commonly used to search for a global minimizer (or maximizer) of an unkonwn or costly fitness function $f: \mathcal{X} \rightarrow \mathbb{R}$: $$\pmb{\alpha}^* = \arg \min_{\pmb{\alpha} \in \mathcal{X}}f(\textbf{x}).$$ In our case, $f(\pmb{\alpha}) = C(\rho(\pmb{\alpha}),\nu)-C(\rho(\pmb{\alpha}),\mu))$ is a fitness function. 

SMBO-based Bayesian optimization optimizes  $\pmb{\alpha}^*$ by prescribing a prior belief over the possible objective functions that are then sequentially refined via Bayesian posterior updating. This updating process is guided by sequentially induced acquisition functions $a_n : \mathcal{X} \rightarrow \mathbb{R}$, which evaluates the utility of candidate points for the next evaluation of $f$ and decides which hyper-parameters to evaluate next.

For the adapted tree-structured parzen estimator approach (TPE), the acquisition functions are guided by the Expected Improvement (EI) criterion. Let $y^*$ be certain threshold value and let $y = f(\pmb{\alpha})$, Expected Improvement utilizes the utility function $$u(y) = \max(0,y^*-y)$$
and evaluates $f$ at the point that, in expectation under some model $M$ of the hyper-parameters, makes the largest magnitude of improvement on the benchmark value $y^*$. The EI acquisition function is thus specified by the following: $$EI_{y^*}(\pmb{\alpha}) := \int_{-\infty}^{y^*} (y^*-y)p_M(y|\pmb{\alpha})dy.$$  

The TPE algorithm models $p(\pmb{\alpha}|y)$ through Bayes' rule $p(\pmb{\alpha}|y) = \frac{p(y|\pmb{\alpha})}{p(y)}$ and using non-parametric densities for the distributions of the configuration prior. Let $l,u$ denote lower bound and upper bound respectively. In our algorithms, we describe the configuration space of target variable $\pmb{\alpha} = (\alpha_1,...,\alpha_n)$ by modeling them as uniform variables over $[l,u]^n$. Using different observations $\{\pmb{\alpha}^{(1)},...,\pmb{\alpha}^{(k)}\}$ in $\text{Unif}([l,u]^n)$, the TPE algorithm substitutes the configuration prior by a truncated Gaussian mixture through a learning process that defines $p(\pmb{\alpha}|y)$ using two densities over the configuration space $\mathcal{X}$: 
$$p(\pmb{\alpha}|y)= 
\begin{cases}
    l(\pmb{\alpha}),& \text{if } y < y^*\\
    g(\pmb{\alpha}),              & \text{otherwise}
\end{cases}$$
Here $l(\pmb{\alpha})$ is the density formed via the ``good" observations $\{\pmb{\alpha}^{(i)}\}$ whose associated loss $f(\pmb{\alpha}^{(i)})<y^*$ whereas $g(\pmb{\alpha})$ is the density formed via the ``bad" remaining observations. The TPE algorithm avoid choosing a minimizing point less than the actual minimizer by choosing $y^*$ to be certain quantile $\gamma$ of the observed $y$ values such that $p(y < y^*) = \gamma$. This parameter $\gamma$ determines the size of the split between the ``good" observation group and the ``bad" ones. Thus, as stated in section 4.1 (Bergstra, Bardenet et. al) the EI-based acquisition function for TPE is proportional the following: 
\begin{align*}
EI_y^{*} \propto (\gamma + \frac{g(\pmb{\alpha})}{l(\pmb{\alpha})}(1-\gamma))^{-1}.
\end{align*}
The proof to this result is contained in the Appendix A.

\begin{algorithm}[H]
\caption{Basic algorithm for Bayesian optimization via TPE}
\textbf{Input:} fitness function $f$, initial model $M$ of hyper-parameter prior distribution, threshold quantile $\gamma$, acquisition function $S$, maximum number of evaluations $N$.
\newline
\textbf{Initialization}: initial acquisition model $M_0 = M$, configuration space $\mathcal{H} = \emptyset$.
\begin{algorithmic}
\For{Trial $1:N$}
    \State Sample a set $H$ of parameter $\pmb{\alpha}$ from $M_{t-1}$.
    \State Evaluate $f(\pmb{\alpha})$ for each one of them.
    \State Split $H$ into good and bad groups (defined above) based on $\gamma$.
\State Update $\pmb{\alpha}^*$ via the following
$$\arg \min_{\pmb{\alpha}}S(M_{t-1},\pmb{\alpha}) := \arg \min_{\pmb{\alpha}} (\gamma + \frac{g(\pmb{\alpha})}{l(\pmb{\alpha})}(1-\gamma))^{-1},$$ where $\pmb{\alpha}$ is selected from the model $M_{t-1}$.
\State Augment $\mathcal{H} = \mathcal{H} \cup (\pmb{\alpha}^*,f(\pmb{\alpha}^*))$.
\State Do a Bayesian update on $M_t$ based on observation history $\mathcal{H}$.
\EndFor
\State \Return $\pmb{\alpha}^*, \mathcal{H}$ 

\end{algorithmic}
\end{algorithm}

Alternatives to TPE include the Gaussian Process Approach (GP), mutual search, and random search. Bergtra et al. [2013] \cite{bergstra2013making} study the performance of each of the Bayesian optimization alternatives and found that while on the convex dataset, both TPE and GP algorithms converge to a validation score (as defined in Bergtra et al. [2013] \cite{bergstra2013making}) of $13\%$, in generalization, TPE's best model had $14.1\%$ error, outperforming GP by $2.6\%$, mutual search by $4.9\%$, and random search by $2.9\%$. In general Bergtra et al. found that models found by the TPE algorithm perform better on the experiment datasets than its alternatives. TPE-based Bayesian hyperparameter optimization is used widely in learning models and shows a steady performance, see Dzikiene et al. [2020] \cite{kapociute2020intent}. Given the typical performance of TPE-based Bayesian optimization algorithm, we choose a TPE-based approach in the numerical implementation. 

The complexity of the presented algorithms depends roughly on three components: the complexity of TPE-based Bayesian optimization framework, the complexity of optimizing the acquisition function, and the complexity of the evaluation of the fitness function. From the discussion in Section 4.1 of Bergstra et al. [2011] \cite{bergstra2011algorithms}, the TPE algorithm maintains sorted lists of observed variables in $\mathcal{H}$, and for each iteration of evaluations, the TPE algorithm scales linearly in the number of variables observed and linearly in the number of variables being optimized. In our case, discounting the costs of numerical evaluations of fitness function, the Indirect Dirichlet algorithm (both histogram and samples-based) has a cost of $O(N|H|p^d)$, whereas the Direct Dirichlet algorithm has a cost proportional to $O(N|H|d)$, where $H$ is the size of sampled hyper-parameters for each trial and $N$ is the maximum number of evaluations.

The minimization of the acquisition function typically uses gradient descent. The time complexity of the gradient descent algorithm for minimizing a function depends on the size of the input and the number of iterations required to converge to a minimum. In general, the time complexity of the gradient descent algorithm is $O(kd)$, where $k$ is the number of iterations required to converge and $d$ is the dimensionality of the input. 

The most expensive step of the computation, however, resides in the evaluation of the fitness function. As mentioned before, our fitness function is $f(\pmb{\alpha}) = C(\rho(\pmb{\alpha}),\nu)-C(\rho(\pmb{\alpha}),\mu))$ for $\pmb{\alpha} \in [l,u]^d$. In both algorithms, we phrase the problem in terms of a linear program, the details of which can be found in previous discussions. Python package \href{https://pythonot.github.io}{POT} implements one of the most efficient exact OT solvers. Our algorithms leverage the regular OT solver \texttt{ot.emd2}. When the input measures are coded as histograms, the OT problem is represented as a linear programming problem and solved with a runtime complexity of $O(n^3)$, where $n$ is the number of bins in the histograms. The computation of the distance matrix between the histograms required by \texttt{ot.emd2} is done in $O(n^2)$ time complexity using a metric such as the Euclidean distance. On the other hand, when the input measures are coded as samples, the computation has a runtime complexity of $O(n^3 \log(n))$ with the \texttt{ot.emd2} solver, although the computation of the distance matrix between the samples remains $O(n^2)$. This is because whereas the histogram-based algorithm uses a fixed grid and thus requires only one calculation of cost matrix, sample-based algorithm calculates the cost matrix in each iteration and scales linearly in the maximum number of evaluations. As shown in the section below, for $d=1,2$ and a medium samples size $n=100$, the histogram-based Indirect Dirichlet method yields consistently a much lower run-time than sample-based Indirect Dirichlet method and Direct Dirichlet method. 

Discounting costs of other computations, for large $n$, the run-time of the Indirect Dirichlet method with histograms is $O(N|H|p^dn^3)$, the Indirect Dirichlet method with samples is $O(N|H|p^d(n^3\log(n))$, and the complexity of the Direct Dirichlet method is $O(N|H|d(n^3\log(n))$. Hence, when $n$ is large and $d$ small (as in our experiments), the Indirect Dirichlet method with histograms performs the best in termf of runtime. However, when $d$ is large, the Direct Dirichlet method will outperform the alternatives. 

For a large scale OT problem in high dimension, POT's Sinkhorn solver \texttt{ot.sinkhorn} has a complexity of $O(n^2)$. Define $$\Omega(\gamma) = \sum_{i,j}\gamma_{i,j}\log(\gamma_{i,j}),$$ then \texttt{ot.sinkhorn} solves the entropic regularization optimal transport problem
\begin{align*}
\min_{\gamma \in \Pi(\mu,\nu)} \{\gamma M + \text{reg}\cdot\Omega(\gamma) \} \hspace{1cm}
\end{align*}
such that 
$\gamma\mathbf{1} = \textbf{a}, \gamma^{\textit{T}}\mathbf{1} = \textbf{b}, \gamma \geq 0$. Generally, \texttt{ot.sinkhorn} speeds up the numerical computation. However, in our case where $n$ and $d$ are small, the improvement in runtime of \texttt{ot.sinkhorn} is not obvious. We thereby choose the regular OT solver \texttt{ot.emd2} for our algorithms.

\subsubsection{Experiments and Results}
We test the performance of algorithms presented above by looking at their ability of detecting the convex order. As in the paper by Wiesel and Zhang, we present three simple cases in which the absence or presence of the convex order is easy to see:

\begin{itemize}
    \item \textbf{Example 1.} $\mu = \mathcal{N}(0,\sigma^2I)$ and $\nu = \mathcal{N}(0,I)$ for $\sigma^2 \in [0,2]$ for $d = 1,2$.
    \item \textbf{Example 2.} $\mu=\frac{1}{2}(\delta_{-1-s}+\delta_{1+s})$ and $\nu = \frac{1}{2}(\delta_{-1}+\delta_1)$ for $s \in [-1,1]$.
    \item \textbf{Example 3.} $\mu = \frac{1}{4}(\delta_{(-1,0)}+\delta_{(1,0)}+\delta_{(0,1)}+\delta_{(0,-1)})$ for $s \in [-1,1]$.
\end{itemize}

Let $$V(\mu,\nu) := \inf_{\rho \in \mathcal{P}^1(\mathbb{R}^d)}(C(\nu,\rho)-C(\mu,\rho)).$$ Then by Theorem 1.2 (Wiesel and Zhang), we have the relationship $$\mu \preccurlyeq_c \nu \text{ iff } V(\mu,\nu) \leq 0.$$ For each example and each pair of $(\mu,\nu)$, we plot $V(\mu,\nu)$ for the Bayesian-optimization-based three methods above: indirect Dirichlet with histograms, indirect Dirichlet with samples, and direct Dirichlet with samples. We refer to the github repository for a more elaborate discussion.

\begin{figure}[H]
\begin{subfigure}{.5\textwidth}
  \centering
  \includegraphics[width=1\linewidth]{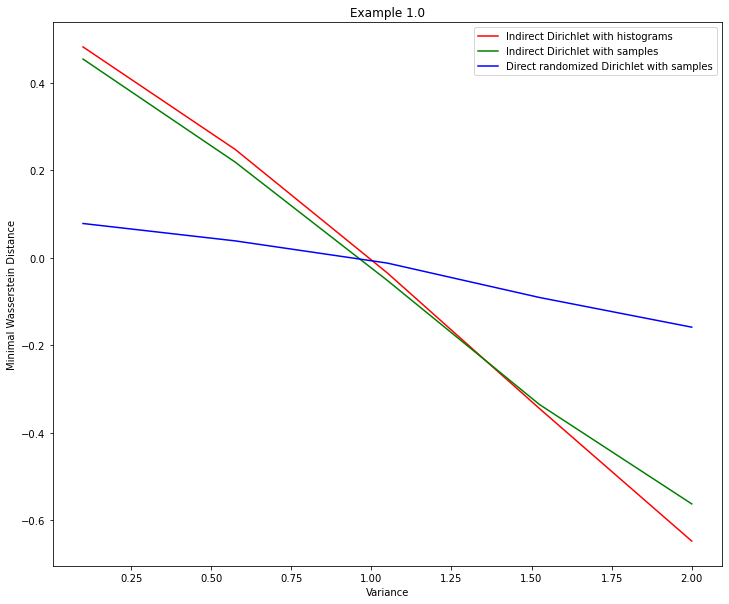}
  \caption{Example 1.0 in 1D}
  \label{fig:sfig12}
\end{subfigure}%
\begin{subfigure}{.5\textwidth}
  \centering
  \includegraphics[width=1\linewidth]{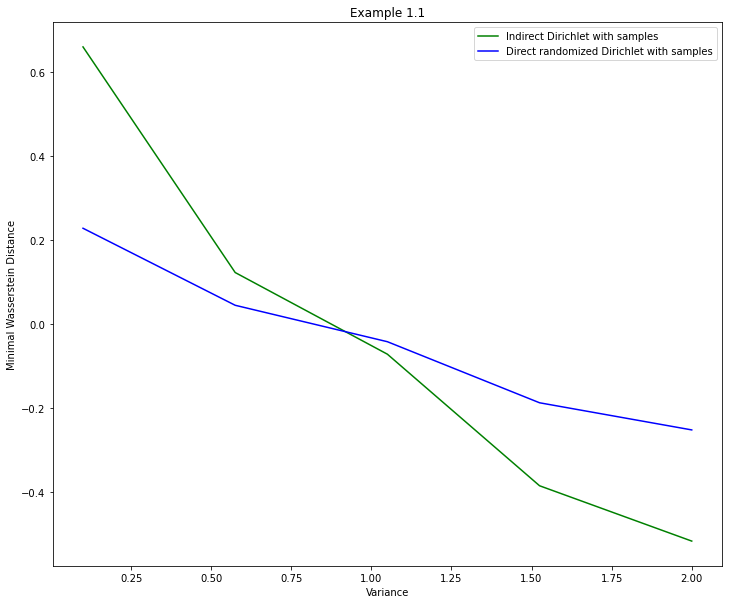}
  \caption{Example 1.0 in 2D}
  \label{fig:sfig24}
\end{subfigure}

\begin{subfigure}{.5\textwidth}
  \centering
  \includegraphics[width=1\linewidth]{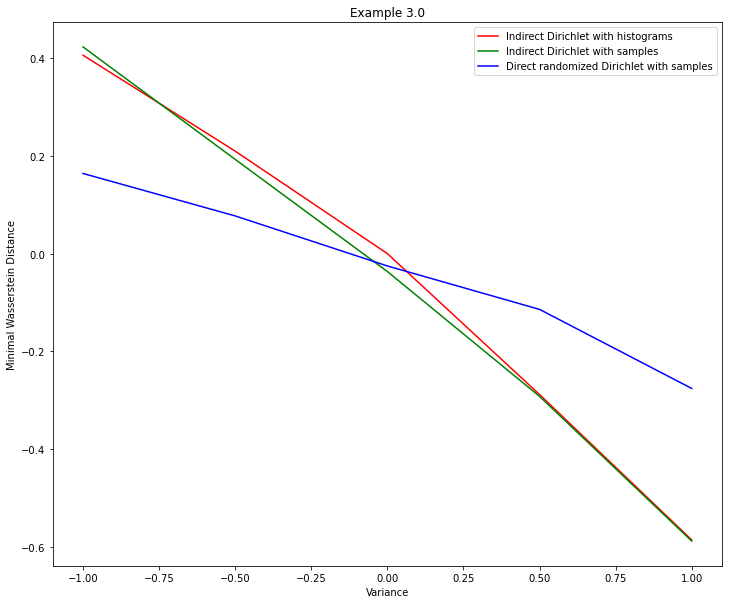}
  \caption{Example 2.0}
  \label{fig:sfig25}
\end{subfigure}
\begin{subfigure}{.5\textwidth}
  \centering
  \includegraphics[width=1\linewidth]{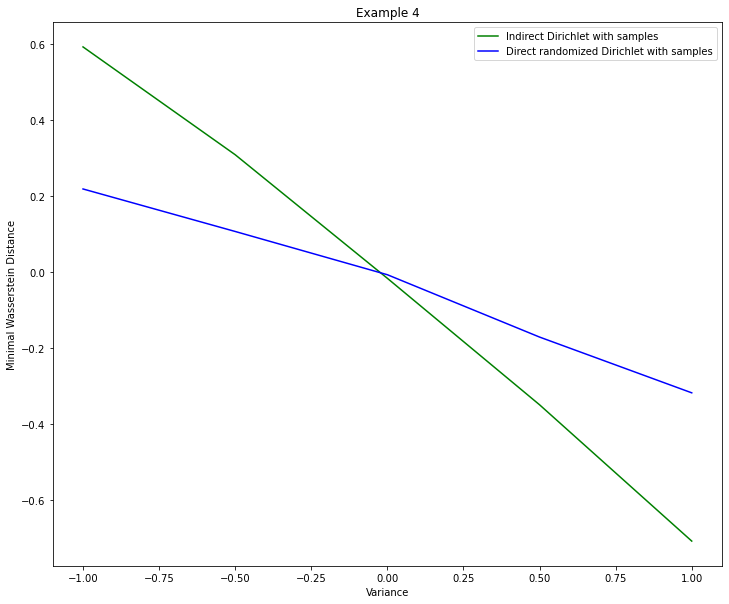}
  \caption{Example 3.0}
  \label{fig:sfig26}
\end{subfigure}
\caption{Values of different estimators of $V(\mu,\nu)$ plotted against $\sigma$ for Examples 1-3. All plots use $N=100$ samples.}
\label{fig:fig1}
\end{figure}

Discounting the numerical errors in sampling, all estimators seem to detect convex order. The direct randomized Dirichlet method is numerically easier to implement; however, it does not seem to explore the $\mathcal{P}^1(\mathbb{R}^d)$-space as well as the other two alternatives. We leave the analysis of this deficiency to another paper. 

As explained in Section 3.4.1, since we only conducted experiments in dimensions we expect that the Indirect Dirichlet method with Histograms to consistently yield the lowest runtime and the runtimes of the other methods to be much higher. This is indeed the case for our experiments. To recall, this is because when working with samples, the weights of the empirical distributions are constant, while the OT cost matrices $\textbf{M}_a$ and $\textbf{M}_b$ in the implementation have to be re-computed in each iteration, which is very costly. For the histogram method, the weights $\rho$ change, while the grid stays constant -- and thus $\textbf{M}_a$ and $\textbf{M}_b$.

\section{Recovering Convex Function \texorpdfstring{$\hat{f}$}{f}}
Fix a $\rho \in \mathcal{P}^1(\mathbb{R}^d)$. The Monge-Kantorivich duality states that
\begin{align*}
C(\mu,\rho) &= \sup_{\pi \in \Pi(\mu,\rho)} \int \langle x,y \rangle \pi(dx,dy) \\
&= \inf_{f \oplus g \geq c} \int fd\mu+\int g d\rho \\
&= \inf_{f \oplus g \geq c, f, g \text{ proper, convex}} \int fd\mu+\int g d\rho
\end{align*}
Take an optimal convex pair $(\hat{f},\hat{g})$ from above for $C(\mu,\rho)$. For the 2-Wasserstein distance, Brenier's theorem in the form of Villani [2003, 2008] Theorem 2.12 \cite{villani2008optimal} gives that the optimal transport map is given by the gradient of a scalar convex function $\hat{f}$: $\rho = \nabla \hat{f}_*\nu$ if $\nu$ is absolutely continuous with the Lebesgue measure. In this section, we present algorithm and graphic example of recovering the convex function $\hat{f}$. 

Given an optimal $\rho \in \mathcal{P}^1(\mathbb{R}^d)$ found through Section 3 in the form of $m$ generated samples, we obtain the optimal transport plan $\pi(\nu,\rho)$. We recover the transport map $\nabla \hat{f}$ through disintegration. Specifically, $\nabla f$ is recovered through taking conditional expectation $\int x \pi_y(dx)$, where $(\pi_y)_{y \in \mathbb{R}^d}$ denotes the conditional probability distribution of $\pi$ with respect to its second marignal $\nu$, i.e. $$\pi_y = \delta_{\nabla \hat{f}(y)} = \pi(x|y).$$ This is a standard technique (see e.g. Deb et al.[2021] \cite{deb2021rates} for details). We then use methods from \texttt{scipy.interpolate} to interpolate the gradient field $\nabla \hat{f}$ in $d-$dimensions. 

To obtain the convex function $\hat{f}$, we consider the following problem. Given a gradient vector field $$\nabla \hat{f}(x_1,...,x_d) := \begin{pmatrix} \frac{\partial \hat{f}}{\partial x_1} \ ... \ \frac{\partial \hat{f}}{\partial x_d} \end{pmatrix},$$ we want to compute the inverse gradient and to recover the scalar field $\hat{f}: \mathbb{R}^d \rightarrow \mathbb{R}$. In the below sections, we introduce algorithms for the computation of the inverse gradient for $d =1,2$.  

\subsection{Recovering \texorpdfstring{$\hat{f}$}{f} in 1D}
Consider a gradient $\nabla \hat{f}(x):\mathbb{R} \rightarrow \mathbb{R}$. Then $\nabla \hat{f}(x) = \frac{d\hat{f}}{dx}$. Hence, $\hat{f}$ is easily recovered up to an additive constant through the Fundamental Theorem of calculus by integrating $\nabla \hat{f}(x)$ with respect to $x$: $$\hat{f}(x) = \int_{c}^x \nabla \hat{f}(y)dy$$ where $c$ is the empirical lower bound. 
\subsubsection{Algorithm for 1D}
To numerically compute the integral above, we use method \texttt{integrate} from \texttt{InterpolatedUnivariateSpline}. To make the pseudo-code clearer, we use list index notation. That is, given a list $l$, $l[i]$ denotes the $i$-th element of $l$. A matrix $M$ in this case is represented by a two-dimensional array, and each element $i,j$ can be accessed through calling index: $M[i][j]$. This leads to the following algorithm.

\begin{algorithm}[H]
\caption{Basic algorithm for recovering $\hat{f}$ in 1D}
\textbf{Input:} $n$ samples drawn from measure $\nu$ (denoted \textbf{b}) and $m$ samples drawn from the optimal measure $\rho$. 

\textbf{Initialization:} weights of measure $\mu$ and $\nu$ are initialized to be uniform, empty list $l$, empty list $l_y$.
\begin{algorithmic}
\State Compute cost matrix $M$ from the samples and calculate optimal transport plan $\pi$ via \texttt{ot.emd}.
\For{iteration $i$ $1:m$}
    \State Append $\frac{\rho[i] \cdot \textbf{G}[i]}{\sum_j^nG[i][j]}$ to $l$.
\EndFor
\For{iteration $j$ $1:n$}
\State Append \texttt{mean}($l^T[j]$) to $l_y$.
\EndFor
\State Interpolate gradient of $\hat{f}$: $$\nabla \hat{f} = \texttt{InterpolatedUnivariateSpline}(\textbf{b},l_y,k=1).$$
\State Obtain $\hat{f}$ by integrating each sample in \textbf{b} from empirical lower bound \texttt{min}(\textbf{b}) via $\nabla \hat{f}$.\texttt{integral}.
\State \Return $\nabla \hat{f}$, $\hat{f}$.
\end{algorithmic}
\end{algorithm}

We offer two methods of smoothing the resulting graph of $\hat{f}$: through building an interpolated spline via \texttt{scipy.interpolate} and using lowess via \texttt{statsmodels.api.nonparametric}. Since both methods yield similar results, we will include only the lowess-based graphs in the below section.

\subsubsection{Experiments and Results for 1D}
We implement Algorithm 7 taking 100 samples from $\mu = \mathcal{N}(0,1)$ and 100 samples from $\nu = \mathcal{N}(0,5)$. We choose a partition grid of $100$ for the computation of $\rho$ and the maximum number of evaluations for Bayesian optimization to be 100. 

\begin{figure}[H]
\begin{subfigure}{.5\textwidth}
  \centering
  \includegraphics[width=1\linewidth]{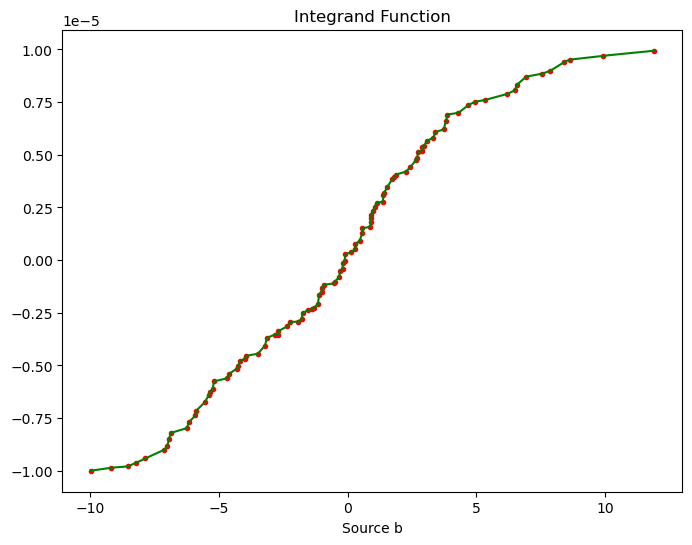}
  \caption{Graph of $\nabla \hat{f}$}
  \label{fig:sfig13}
\end{subfigure}%
\begin{subfigure}{.5\textwidth}
  \centering
  \includegraphics[width=1\linewidth]{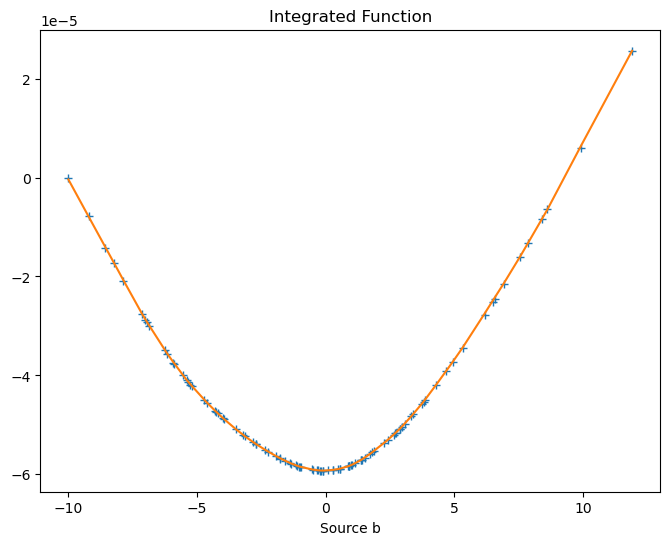}
  \caption{Graph of $\hat{f}$}
  \label{fig:sfig27}
\end{subfigure}
\caption{Graph of $\nabla \hat{f}$ and $\hat{f}$ for example 4.1.2}
\label{fig:fig2}
\end{figure}
As expected, the integrand $\nabla \hat{f}$ appears to be affine and the recovered function $\hat{f}$ convex and follow the form $\hat{f}(x) = ax^2$ for some scalar $a \in \mathbb{R}$.

\subsection{Recovering \texorpdfstring{$\hat{f}$}{f} in 2D}
Consider the gradient $\nabla \hat{f}(x,y) :=\begin{pmatrix} \frac{\partial \hat{f}}{\partial x} \ \frac{\partial \hat{f}}{\partial y} \end{pmatrix}$. In implementation, we observe an empirical vector field $g$ of $\hat{f}$ on some bounded and connected domain $\Omega \in \mathbb{R}^n$. We expect that the real gradient field $\nabla \hat{f}(x,y)$ to be near $g$. Then numerically, to recover $\hat{f}(x,y)$ with its gradient field near some given observed vector field $g$, we consider the minimization of $\|\nabla \hat{f} - g\|$ for some choice of norm and over some class of scalar functions mapping from $\mathbb{R}^2$ to $\mathbb{R}$. Consider the $L^2(\Omega)$-norm of the pointwise Euclidian norm, we have the below minimization problem: $$\arg \min_{\hat{f}} \|\nabla \hat{f} - g \|^2 = \int_{\Omega}|\nabla \hat{f}(\textbf{x})-g(\textbf{x})|^2d\textbf{x},$$ where $\Omega \subset \mathbb{R}^n$ is bounded and connected. From the calculus of variation, it is well established that (see Farbebäck[2007] \cite{4409176}, Song et al. [1995] \cite{song1995phase}) the inverse gradient $\hat{f}(x,y)$ is the solution to a Poisson equation with inhomogeneous Neumann boundary conditions,
$$\begin{cases}
    \nabla^2\hat{f}(\textbf{x}) = \nabla \cdot g(\textbf{x}),& \forall \textbf{x} \in \Omega\\
    \nabla \hat{f}(\textbf{x}) \cdot \textbf{n} = \textbf{n} \cdot g(\textbf{x}), & \forall\textbf{x} \in \partial \Omega,
\end{cases}$$
where $\partial \Omega$ is the boundary of $\Omega$, $\textbf{n}$ is a normalized and outwards directed normal vector to the boundary, and $\nabla^2$ is the Laplace operator. By Theorem 3 in Evans [2010] \cite{evans2010partial}, the necessary condition for the existence of a solution to the Neumann problem, as in our case, is that $$\int_\Omega -\nabla g(\textbf{x}) d\textbf{x} = - \int_{\partial \Omega} g dS,$$ where $ - \int_{\partial \Omega} g dS$ denotes the line integral with respect to the boundary $\partial \Omega$. Note that this is satisfied because of the Generalized Stoke's theorem, which states that $$\int_\Omega d\beta = \oint_{\partial \Omega} \beta dS$$ for any differential form $\beta$ and domain $\Omega$. Thus, we are guaranteed to find a solution.

\subsubsection{Algorithm for 2D}
Numerical methods to solve the Poisson equation with Neumann boundary conditions are well developed. Farbebäck et al. [2007] \cite{4409176}, for example, developed efficient multigrid-based solver to solve the above Poisson equation with an inhomogeneous Neumann boundary conditions over irregular domains. In our implementation, we work with python package \href{https://fenicsproject.org/olddocs/dolfin/1.4.0/python/demo/documented/neumann-poisson/python/documentation.html#:~:text=For%20a%20domain%20%CE%A9%E2%8A%82,c%20by%20the%20above%20equations.}{\texttt{dolfin}}.

\begin{algorithm}[H]
\caption{Basic algorithm for recovering $\hat{f}$ in 2D}
\textbf{Input:} $n$ samples drawn from measure $\nu$ (denoted \textbf{b}) and $m$ samples drawn from the optimal measure $\rho$. 

\textbf{Initialization:} weights of measure $\mu$ and $\nu$ are initialized to be uniform, empty list $l$, empty list $l_y$, \texttt{dolfin} function space with Lagrange multiplier $W$. 
\begin{algorithmic}
\State Compute cost matrix $\textbf{M}$ from $\nu$ to $\rho$ and calculate optimal transport plan $\textbf{G}$ via \texttt{ot.emd}.
\For{iteration $i$ $1:n$}
    \State  Append $\frac{\rho \cdot \textbf{G}[i]}{\sum_j^nG[i][j]}$ to $l$.
\EndFor
\State Initialize gradient domain $\Omega$ based on observed gradient points. Set empirically observed gradient field $g = l.$
\State Set trial function (represents the unknown $\hat{f}$) to be $(\hat{f}, \hat{c})$ and test function to be $(f,c)$, where $\hat{c},c$ represent the respective constant terms. Solve variational problem using \texttt{dolfin}:
$$\nabla \hat{f} \cdot \nabla f + \hat{c}f +c\hat{f}dx = \nabla gfdx+(\textbf{n}\cdot g)fds.$$ 
\State \Return $\nabla \hat{f}$, $\hat{f}$.
\end{algorithmic}
\end{algorithm}

For graphing the resulting function $\hat{f}: \mathbb{R}^2 \rightarrow \mathbb{R}$ on irregular domains consisting of the empirical samples from $\nu$, we use python package \texttt{matplotlib.tri}. \texttt{matplotlib.tri} interpolates and generates a 3D surface by using a Delaunay triangulation. The resulting graph is presented in the section below.

\subsubsection{Experiments and Results for 2D}
We implement algorithm 8 on the same example in section 4.2.1 in 2D. That is, $\mu = \mathcal{N}(0,5I)$ and $\nu = \mathcal{N}(0,I)$. We choose a partitioned grid of $50$ for the computation of $\rho$ and the maximum number of evaluations for Bayesian optimization to be $100$. 

\begin{figure}[H]
\begin{subfigure}{.5\textwidth}
  \centering
  \includegraphics[width=1\linewidth]{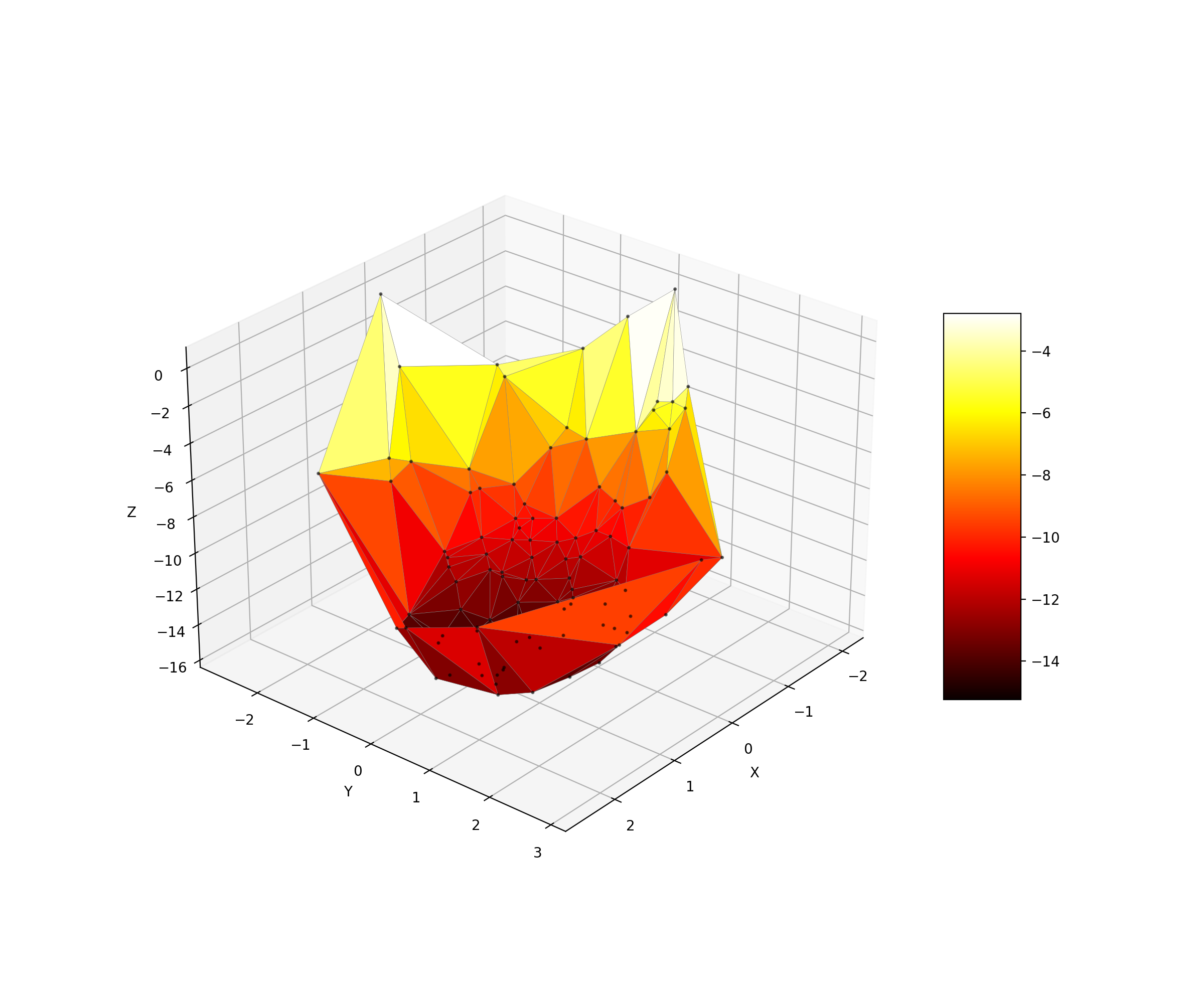}
  \caption{x-y-z view}
  \label{fig:sfig14}
\end{subfigure}%
\begin{subfigure}{.5\textwidth}
  \centering
  \includegraphics[width=1\linewidth]{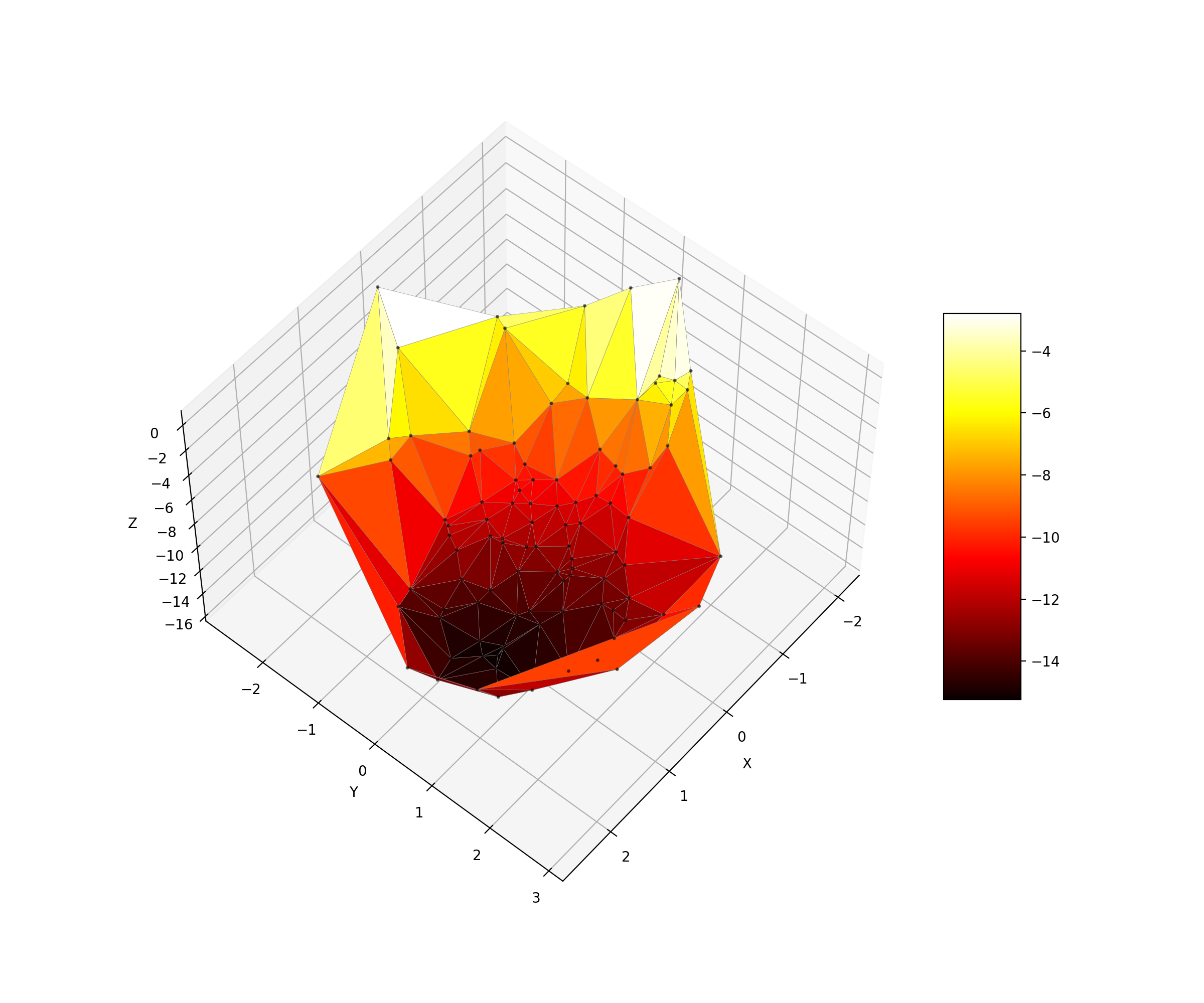}
  \caption{x-y view}
  \label{fig:sfig28}
\end{subfigure}
\begin{subfigure}{.5\textwidth}
  \centering
  \includegraphics[width=1\linewidth]{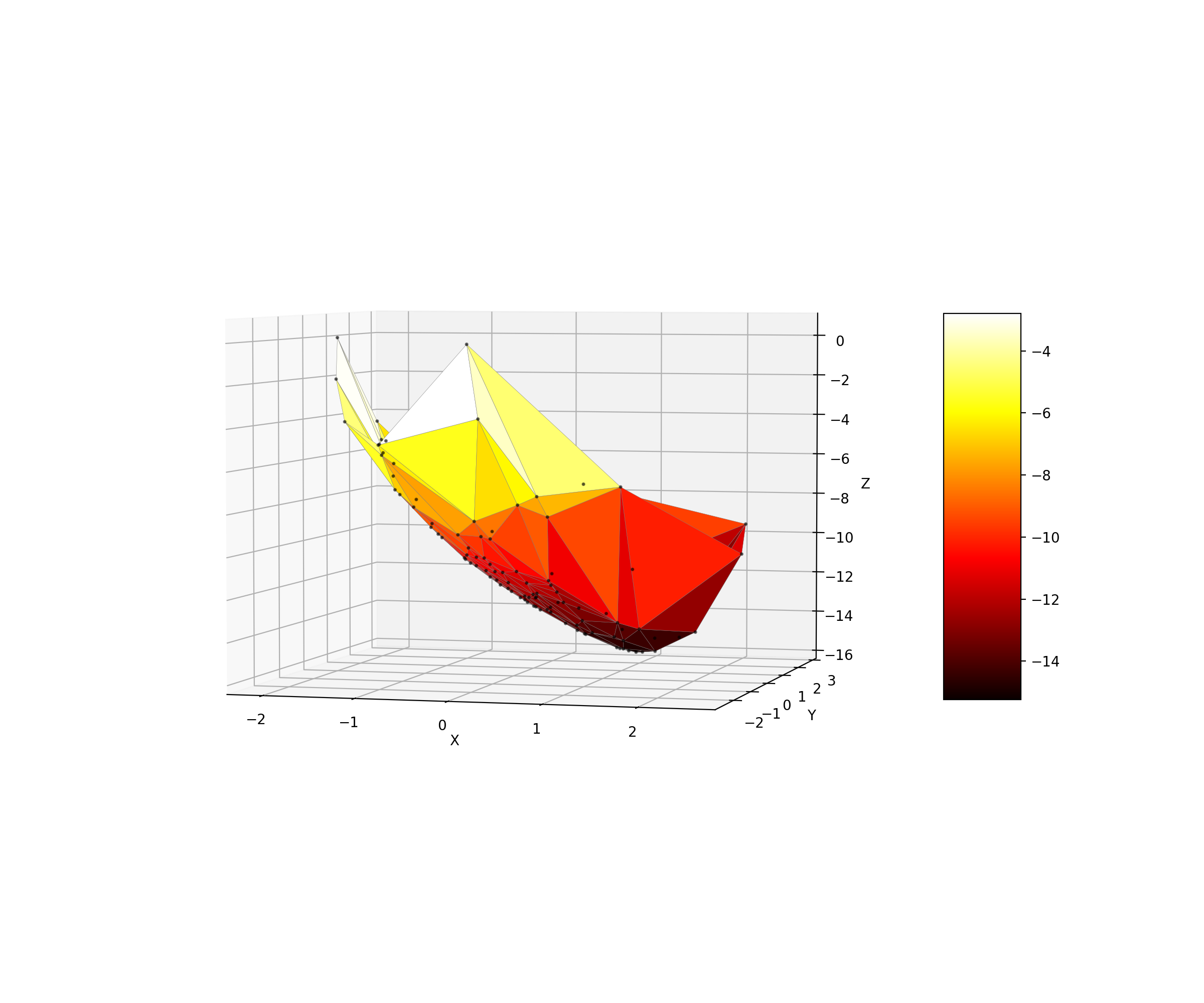}
  \caption{x-z view}
  \label{fig:sfig29}
\end{subfigure}
\begin{subfigure}{.5\textwidth}
  \centering
  \includegraphics[width=1\linewidth]{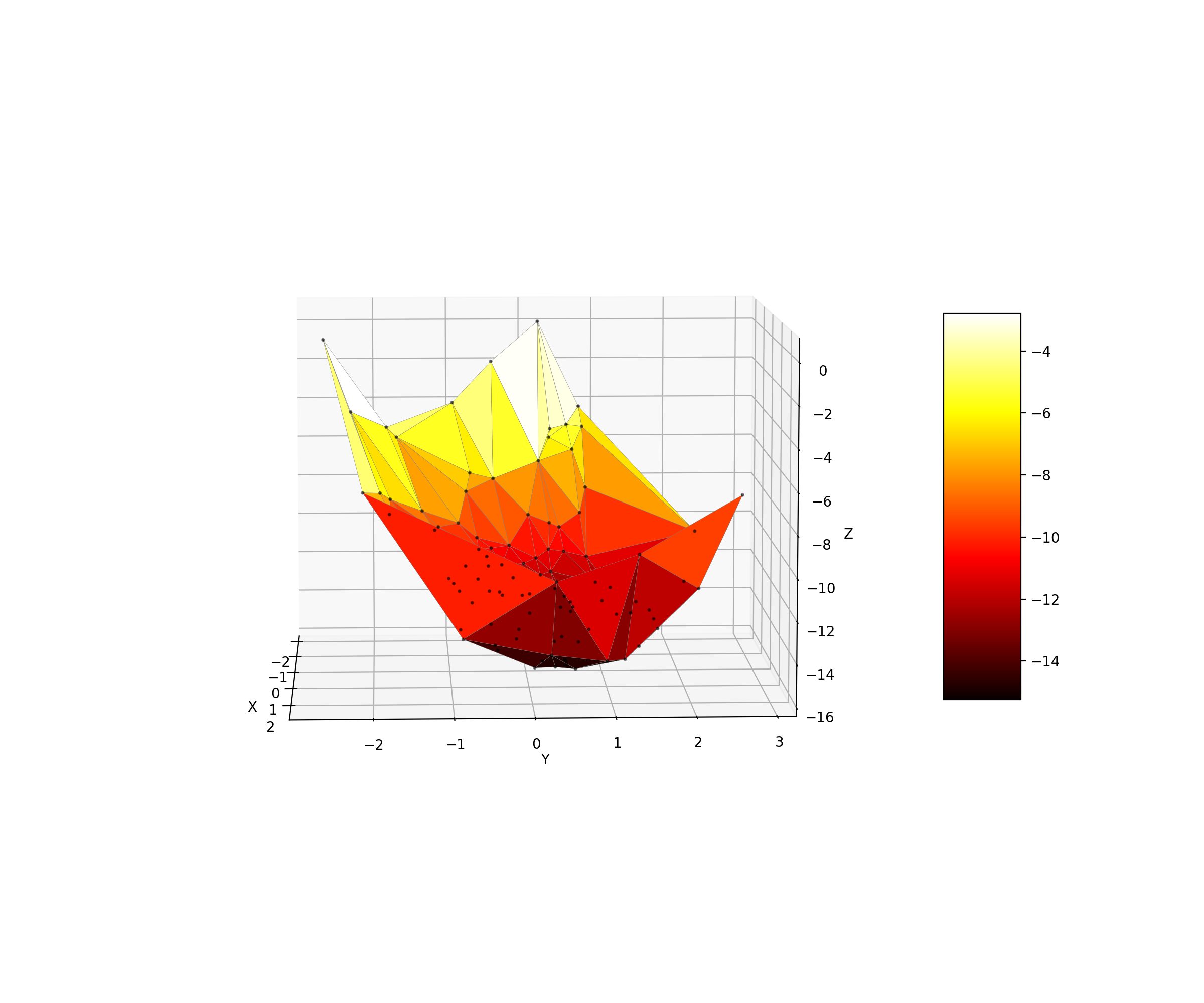}
  \caption{y-z view}
  \label{fig:sfig23}
\end{subfigure}
\caption{Graph of $\hat{f}$ for example 4.2.2}
\label{fig:fig3}
\end{figure}

As expected, the resulting $\hat{f}$ is convex and has the form of a parabola, analogous to the example in 1D.

\section{Model-free arbitrage strategy via \texorpdfstring{$\hat{f}$}{f}}

Recall the market setup established in Section 2 and consider the $d$ financial assets with price processes $(S_t)_{t \geq 0} = (S^1_t,...,S^d_t)_{t \geq 0}$. Fix two maturities $T_1<T_2$. Assume that call options with these maturities are traded at all strikes, then by the Breeden-Litzenberger formula (introduced in Section 1.2), the prices of the call options determine the distribution of $S_{T_1}$ and $S_{T_2}$ under any equivalent martingale measures. Let us denote the distributions of $S_{T_1}$ and $S_{T_2}$ by $\mu$ and $\nu$ respectively. We then construct model-free arbitrage strategy using the convex function recovered in Section 3 between $\mu$ and $\nu$, where \textit{model-independent arbitrage} is understood w.r.t semi-static trading strategies. To establish the result, we recall a few definitions.

\begin{definition}[\textit{Martingale coupling}]
Let $\mathcal{M}(\mu,\nu)$ denote the set of all martingale laws $\pi$ with marginals $\mu,\nu$. That is, probability measure $\pi$ such that $$x \sim \mu, \ \ y \sim \nu, \ \ \mathbb{E}^\pi[y|x] = x.$$
We call $\mathcal{M}(\mu,\nu)$ the set of all martingale couplings between $\mu,\nu$.
\end{definition}

\begin{definition}[\textit{Semi-static trading strategy}]
 A trading strategy $\Delta = (\Delta_t)_{t=0}^{T-1}$ is a semi-static strategy if it consists of a static portfolio in finitely many options whose prices are known at time zero, and a dynamic, self-financing strategy in the underlying assets.
\end{definition}

 As mentioned above, we define model-independent arbitrage w.r.t semi-static trading strategies. We follow definition 2.1 from Davis and Hobson [2007] \cite{davis_hobson_2007}.

\begin{definition}[\textit{Model-independent arbitrage}]
A model-independent arbitrage exists when there exists a semi-static portfolio with zero initial value and with strictly positive value at the terminal date across all possible scenarios. 
\end{definition}

Our set up for model-independent arbitrage is motivated by the above definitions. Suppose that trading is only allowed at $0,T_1,T_2$, we give the following adapted version of model independent arbitrage. This is stated in Definition 5.1 by Wiesel and Zhang [2023] \cite{wiesel2023optimal}.

\begin{definition}[\textit{Adapted model-independent arbitrage}]
Consider trading strategy $\Delta$. The triple of Borel measurable functions $(u_1,u_2,\Delta)$ forms a model-independent arbitrage if $u_1 \in L^1(\mu)$ and $u_2 \in L^1(\nu)$ such that $$u_1(x) - \int u_1d\mu +u_2(y)-\int u_2d\nu + \Delta(x)(y-x) >0, \ \ \forall (x,y) \in \mathbb{R}^d \times \mathbb{R}^d.$$
If no such strategies exist, then we call the market free of model-independent arbitrage.
\end{definition}

Building on Theorem 3.4 from Guyon et al. [2017] \cite{guyon2017bounds}, we present the following theorem connecting model-independent arbitages and convex order between $\mu,\nu$. This is stated in Theorem 5.2 and its proof in Wiesel and Zhang [2023] \cite{wiesel2023optimal}.

\begin{theorem}
The following are equivalent:
\begin{enumerate}
\item The market is free of model-independent arbitrage.
\item $\mathcal{M}(\mu,\nu) \neq \emptyset$.
\item $\mu \preceq_c \nu$.
\end{enumerate}
In particular, if $\mu \preceq_c \nu$, then there exists a convex function $f$, such that the triple $(-f(x),f(y),g(x))$ is a model-independent arbitrage. Here $g$ is a measurable selector of the subdifferential of $f$.
\end{theorem}
We give a proof of Theorem 5.0.1.
\begin{proof}[\textit{Proof of Theorem 5.0.1}]
(ii) $\leftrightarrow$ (iii) is stated in Strassen's theorem. If $\mu \not \preceq_c \nu$, then by definition of convex order there exists a convex function $f$ such that $$\int f d\mu > \int f d\nu.$$ On the other hand, $f$ is convex and thus satisfies $$f(y)-f(x) \geq g(x)(y-x) \ \ \forall (x,y) \in \mathbb{R}^d \times \mathbb{R}^d.$$ Combining the two equations above shows that $(-f(x),g(y),g(x))$ is a model-independent arbitrage and thus (i) $\rightarrow$ (ii). It remains to show that (ii) $\rightarrow$ (i). Taking expectations in the inequality $$u_1(x)-\int u_1d\mu +u_2(y)-\int u_2 d\nu + \Delta(x)(y-x) >0, \ \ \forall(x,y) \in \mathbb{R}^d \times \mathbb{R}^d$$ under any martingale measure with marginals $\mu,\nu$ leads to a contradiction. This concludes the proof.
\end{proof}

In particular, a model-independent arbitrage strategy is given by the triple $(-\hat{f}(x), \hat{f}(y), \nabla \hat{f}(x)).$ This is commonly known as the calendar spread, which is a trading strategy that involves buying and selling options or futures contracts with different expiration dates but the same underlying asset and strike price. In particular, once $\rho$ is found numerically (see section 3), we need only to estimate $\nabla f(x)$ from the optimizing transport plan of $C(\rho,\nu)$ in order to obtain an explicitly arbitrage strategy (see section 4).

\section{Arbitrage between Portfolios}
Another exploitation of the absence of convex order in the market for arbitrage opportunities can be found in the arbitrage between two Functionally Generated Portfolios. In this section, we introduce Functionally Generated Portfolio in Stochastic Portfolio Theory and propose a tentative link between convex order and Functionally Generated Portfolios. Specifically, we propose that the ``convex order" at play in the arbitrage strategy is a different measurement of the volatility of market, captured by the constructed gamma process in Fernholz, Karatzas et al. [2017] \cite{fernholz2016volatility}. We begin by introducing the Functionally Generated Portfolio in Stochastic Portfolio Theory and setting up relevant definitions. Note that the market setup and notation can be different from the previous sections.

\subsection{Functionally generated portfolio}
Functionally Generated Portfolio is the method of constructing trading strategies through functional generation. This is an important construction in Stochastic Portfolio Theory established by Robert Fernholz. Since then, Karatzas and Karadaras (see \textit{Portfolio theory and Arbitrage} [2021] \cite{karatzas2012portfolio}) and Karatzas and Ruf [2017] \cite{karatzas2016trading} etc. further developed the theory systematically and simplified the Functionally Generated Portfolio method by proposing additive and multiplicative functional generation methods. We begin by setting up the relevant market model considered in the Functionally Generated Portfolio framework by defining key concepts.

\subsubsection{Market model}
On a given probability space ($\Omega, \mathcal{F}, \mathbb{P}$), endowed with a right continuous filtration $(\mathcal{F})_{t \geq 0}$ that satisfies $\mathcal{F}(0) = \{\emptyset, \Omega\}$, we consider a stochastic process $S = (S_1,...,S_d)$ of continuous non-negative semi-martingales with strictly positive initial values for each component. Now consider a frictionless equity market composed of fixed $d$ assets and a money market. Let the vector process $S_i$ represents the price of asset $i$ where $i = 1,2,...,d$. Define the process 
$$\Sigma(t) := S_1(t)+...+S_d(t), \ \ t \geq 0.$$
Then this process is interpreted as the market capitalizations of a fixed number $d \geq 2$ of companies in the equity market. We further impose the constraint that the total capitalization $\Sigma$ of the equity market satisfies
$$\mathbb{P}[\Sigma(t) > 0, \ \ \forall t \geq 0] = 1.$$ Now define the vector process $\pi = (\pi_1,...,\pi_d)^T$ that consists of the respective market one $d$ companies' relative market weight processes
$$\pi_i(t) := \frac{S_i(t)}{\Sigma(t)} = \frac{S_i(t)}{S_1(t)+...+S_d(t)}, \ \ t\geq 0, i = 1,...,d.$$ Then each components of $\pi$ are continuous, nonnegative semimartingales, each of them takes values in the unit interval $[0,1]$, and they satisfy $\pi_1+...+\pi_d = 1$. More fromally, the continuous $d-$dimensional semi-martingale $\pi(t) = (\pi_1(t),...\pi_d(t))$ takes values in the lateral face $$\Delta^d :=\{(x_1,...,x_d) \in [0,1]^d : \sum_{i=1}^d x_i = 1\} \subset \mathbb{H}^d$$ of the unit simplex, where $\mathbb{H}^d$ denotes the hyperplane $$\mathbb{H}^d:= \{(x_1,...,x_d) \in \mathbb{R}^d:  \sum_{i=1}^d x_i = 1\}.$$ We assume that the initial weights $\pi^i(0)$ for $i \in \{1,..,d\}$ are non-negative.

\subsubsection{Trading strategies}
We consider predictable processes $\psi = (\psi_1,...,\psi_d)$ and $\phi = (\phi_1,...,\phi_d)$ with values in $\mathbb{R}^d$. We define $\psi_i$ and $\phi_i$ as the number of shares of stocks of company $i = 1,...,d$ at time $t\geq0$. Then the associated wealth processes of the investment under strategy $\psi$ and $\phi$ are
$$V^\psi(\cdot ; S) : = \sum_{i = 1}^d \psi_iS_i$$
and 
$$V^\phi(\cdot ; S) : = \sum_{i = 1}^d \phi_iS_i.$$ We now give the definition of trading strategies in our setting.

\begin{definition}[\textit{Trading strategies}]
Define $\mathcal{L}(\cdot)$ to be the set of admissible trading strategies (or portfolios) of the equity market $S$ (or market weights $\pi$). Then for $\psi$ and $\phi$ to be admissible trading strategies with respect to equity market $S$ if they are ``self-financed", i.e.
$$V^{\psi}(\cdot; S) - V^\psi(0;S) = \int_{0}^{\cdot} \sum_{i = 1}^d \psi_i(t)dS_i(t).$$
$$V^{\phi}(\cdot; S) - V^\phi(0;S) = \int_{0}^{\cdot} \sum_{i = 1}^d \phi_i(t)dS_i(t).$$
\end{definition}

In this case, we write $\psi, \phi \in \mathcal{L}(S)$. Moreover, it is important to note that by Proposition 2.3 in Karatzas and Ruf [2017] \cite{karatzas2016trading}, $\mathcal{L}(S) = \mathcal{L}(\pi)$, i.e. a $\mathbb{R}^d$-valued previsible process $\psi$ is a trading strategy with respect to the semimartingale $S$ iff it is a trading strategy with respect to the market weight semimartingale $\pi$ as defined previously. We write in this case
$$V^\psi (\cdot; S) = \Sigma(\cdot) V^\psi(\cdot; \pi)$$
and 
$$V^\phi (\cdot; S) = \Sigma(\cdot) V^\phi(\cdot; \pi).$$

\subsubsection{Proportional Investment}
By our setup, for an admissible trading strategy $v = (v_1,...,v_d) \in \mathcal{L}(S)$, the associated wealth process 
$$V^v (\cdot; S) = V^v(0;S)+\int_{0}^{\cdot} \sum_{i = 1}^d v_i(t)dS_i(t)$$
is a numeraire by definition. Since $V^\psi (\cdot; S), V^\phi (\cdot; S) >0$, we can define new predictable, vector-valued processes $\mu = (\mu_1,...,\mu_d)^T \in [0,1]^d$ and $\nu = (\nu_1,...,\nu_d)^T \in [0,1]^d$ with respect to $\psi,\phi$ in the following manner:
$$\mu_i(t) := \frac{S_i(t)\psi_i(t)}{V^\psi(t;S)}$$
$$\nu_i(t) := \frac{S_i(t)\phi_i(t)}{V^\phi(t;S)}$$
for each $i \in \{1,...,d\}$.We interpret $\mu$ and $\nu$ as two trading portfolios in the $d-$asset market, where $\mu_i$ and $\nu_i$ each represents the proportion of current wealth at time $t\geq0$ invested in the stock of company $i = 1,...,d$.

\subsubsection{Additively generated portfolios}
Fernholtz, Karatzas, and Ruf [2018] \cite{fernholz2016volatility} establish a special class of trading strategies for which the representation in 6.1.2 takes a very simple form. These trading strategies are classified in two groups: the additive functional generation and the multiplicative functional generation. In this paper, we concern ourselves with only the additive Functionally Generated Portfolios. To introduce this concept, we begin with a definition of regular function following Definition 3.1 in Karatzas and Ruf [2017] \cite{karatzas2016trading}.
\begin{definition}[\textit{Regular function}]
A continuous mapping $G: \Delta^d \rightarrow \mathbb{R}$ is called a regular function with respect to some vector process $(\pi(t))_{t\geq0}$ of relative market weights, if the process $G(\pi(t))$ is a semimartingale of the form 
$$G(\pi(t)) = G(\pi(0))+\int_0^t \sum_{i=1}^d D_iG(\pi(s))d\pi_i(s) -\Gamma^G(t)$$ for some measurable function $DG = (D_1G,...,D_dG): \Delta^d \rightarrow \mathbb{R}^d$ and a continuous, adapted process $(\Gamma^G(t))_{t\geq0}$ of finite variation on compact time intervals. If the continuous function $G$ can be extended to a twice continuously differentiable function, elementary stochastic calculus then expresses the gamma process as: $$\Gamma^G(t) = \frac{1}{2}\sum_{i=1}^d \sum_{j=1}^d \int_0^t D_{ij}^2G(\pi(s))d\langle \pi_i,\pi_j\rangle (s)$$ using the notation $$D_iG = \frac{\partial G}{\partial x_i}, \ \  D^2_{ij} = \frac{\partial^2 G}{\partial x_i\partial x_j}.$$ If the function $G$ is concave, then the process $\Gamma^G(t)$ is non-decreasing.

\end{definition}

\begin{definition}[\textit{Lyapunov Function}]
A regular function $G$ is a Lyapunov function if the process $\Gamma^G(\cdot)$ in Definition 6.1.2 is nondecreasing.
\end{definition}

Whereas Fernholtz, Karatzas, and Ruf [2018] developed both the notion of the additively generated portfolios and that of the multiplicatively generated portfolios, in this paper we restrict ourselves to the discussion of only the additivley generated portfolios. We follow Definition 4.1 in Karatzas and Ruf [2017] \cite{karatzas2016trading}
\begin{definition}[\textit{Additively generated portfolios}]
For any given regular function $G: \Delta^d \rightarrow \mathbb{R}$ w.r.t. the vector process $(\pi(t))_{t \geq 0}$ of market weights, we consider $d$-dimensional vector $(\phi^G(t))_{t\geq 0} = (\phi^G_1(t),...,\phi^G_d(t))_{t\geq 0}$ with components $$\phi_i^G(t) := D_iG(\pi(t))+\Gamma^G(t)+G(\pi(t))-\sum_{j=1}^d \pi_j(t)D_jG(\pi(t)), \ \ \text{for }i = 1,...,d.$$ This process is said to be a portfolio \textit{additively generated by $G$}. 

By Proposition 4.3 of Karatzas and Ruf [2017] \cite{karatzas2016trading}, the additively generated portfolio $(\pi(t))_{t \geq 0}$ from Definition 6.1.2 is a trading strategy in the sense of Definition 6.1.1  and has the below \textit{relative value process}: $$V^\phi(t) = G(\pi(t))+\Gamma^G(t).$$

\end{definition}

\subsection{Functionally generated arbitrage}
For previous works on functionally generated portfolios (see Karatzas and Ruf [2017] \cite{karatzas2016trading} and Fernholz, Karatzas, and Ruf [2018] \cite{fernholz2016volatility}), arbitrage and trading strategies of functionally generated portfolios are studied with respect to the benchmark market portfolio. In this section, we consider a generalization of the theory by considering arbitrage of a functionally generated portfolio with respect to another functionally generated portfolio. 

\subsubsection{Relative arbitrage}
We begin by defining relative arbitrage of a functionally generated portfolio (or its respective trading strategy) with respect to another portfolio. For this, we adapt and extend Definition 4.1 in Fernholtz et al. [2018] \cite{fernholz2016volatility}.

\begin{definition}[\textit{Relative Arbitrage}]
Let us fix a real number $T>0$ and a frictionless equity market $S$ with $d$ stocks. Let the process $\pi$ represent the relative market weight process. Given trading strategies $\psi,\phi \in \mathcal{L}(S)$, we define their respective trading portfolios $\mu := (\mu_1,...,\mu_d)$ and $\nu := (\nu_1,...,\nu_d)$ by
$$\mu_i(t) := \frac{S_i(t)\psi_i(t)}{V^{\psi}(t;S)}$$
$$\nu_i(t) := \frac{S_i(t)\phi_i(t)}{V^{\phi}(t;S)}$$
for $i = 1,2,...,d$ (see Section 6.1.3). Then portfolio $\mu$ is an arbitrage relative to portfolio $\nu$ over the time horizon $[0,T]$ if we have 
$$V^{\psi}(t;\pi),V^{\phi}(t;\pi) \geq 0, \ \ \forall t \in [0,T]; \ \ \ \ \ V^\psi(0;\pi) = V^{\phi}(0;\pi) = 1$$
along with
$$\mathbb{P}[V^\psi(T;\pi)-V^\phi(T;\pi) \geq 0] =1$$
and
$$\mathbb{P}[V^\psi(T;\pi)-V^\phi(T;\pi) > 0] > 0.$$
Moreover, whenever a given portfolio $\mu$ satisfies these conditions, and if the last probability is not just positive but actually equal to 1, that is, if 
$$\mathbb{P}[V^\psi(T;\pi)-V^\phi(T;\pi) > 0] =1,$$
we say that the portfolio $\mu$ generated by $\psi$ is a strong arbitrage relative to the portfolio $\nu$ generated by $\phi$ over the time horizon $[0,T]$. Equivalently, in terms of trading strategies, 
we say that trading strategy $\psi$ is a strong arbitrage relative to the trading strategy $\phi$ over the time horizon $[0,T]$. 

\end{definition}

\subsubsection{Relative arbitrage between additively generated portfolios}
In this section, we examine relative arbitrage between two additively generated portfolios. This is a generalization of the result in Theorem 4.3 by Karatzas and Ruf [2017] \cite{karatzas2016trading}. We leave the analysis of multiplicatively generated portfolio to future work.

\begin{proposition}
Let continuous functions $G_1,G_2 : \Delta^d \rightarrow [0,\infty)$ be Lyapunov functions with $G_1(\pi(0)), G_2(\pi(0))>0$. Further impose $G_1,G_2$ to be non-negative to ensure non-negative wealth processes generated by the trading strategies derived from $G_1,G_2$. Assume that $G_2(\pi(t))$ is bounded above by some constant $C$ for all $t \geq 0$. Let $\Gamma^{G_k}$ denote the Gamma process generated by functions $G_k$ for $k = 1,2$. If (continuous, adapted) process
$$\kappa (\cdot) := \Gamma^{G_1}(\cdot) - \Gamma^{G_2}(\cdot)$$ 
is such that
$$\kappa(t) > \eta t $$ for some constant $\eta >0$, then there exists a time horizon over which the trading strategy $\psi = (\psi_1,...,\psi_d)$ additively generated by $G_1$ is a strong arbitrage relative to the trading strategy $\phi = (\phi_1,...,\phi_d)$ additively generated by $G_2$. 
\end{proposition}

\begin{proof}
The following proof simplifies the notation for relative value process notation $V(\cdot;\pi)$ to $V(\cdot)$. Let trading strategies $\psi,\phi$ be additively generated with respect to functions $G_1,G_2$. By Proposition 4.3 in ``Lyapunov functions", 
$$V^\psi (t) = G_1(\pi(t))+\Gamma^{G_1}(t)$$
$$V^\phi (t) = G_2(\pi(t))+\Gamma^{G_2}(t)$$
for $t\geq 0$.
It follows that given a fixed $T>0$, 
\begin{equation}
\begin{split}
V^\psi (T)-V^\phi (T) & = (\Gamma^{G_1}(T)-\Gamma^{G_2}(T))+(G_1(\pi(T))-G_2(\pi(T)))\\
 & \geq (\Gamma^{G_1}(T)-\Gamma^{G_2}(T))-G_2(\pi(T))
\end{split}
\end{equation}
Here we have used the bound $G_1 \geq 0$. From this, we can derive a stronger condition for strong relative arbitrage. Specifically, if
$$\mathbb{P}[\Gamma^{G_1}(T)-\Gamma^{G_2}(T) > G_2(\pi(T))] = 1$$
then
$$\mathbb{P}[V^\psi (T)-V^\phi (T) > 0] = 1$$
holds almost surely. Recall that $\kappa(t) = \Gamma^{G_1}(t)-\Gamma^{G_2}(t) > \eta t$ for some $\eta >0$ and $\forall t \geq 0$. Choose $T^*>0$ to be such that 
$$T^* > \frac{C}{\eta},$$
it follows that
$$\Gamma^{G_1}(T^*)-\Gamma^{G_2}(T^*) > \eta T^* > C \geq G(\pi(T^*)).$$
Then for all $T \geq T^*$, we have that
$$\Gamma^{G_1}(T)-\Gamma^{G_2}(T)> \eta T > \eta T^* > C \geq G(\pi(T)).$$
Therefore, the additively generated strategy $\psi = (\psi_1,...,\psi_d)^T$ is a strong arbitrage relative to the additively generated strategy $\phi = (\phi_1,...,\phi_d)^T$ over every time horizon $[0,T]$ with $T \geq T^*$. 
\end{proof}

\begin{remark}
Note that by the definition of regular functions, the trading strategy additively generated by a regular function $G$ need not be unique. This is because $DG$ is not necessarily unique. This may result in different wealth process additively generated by the same regular function $G$. We note that our results in Proposition 1.3 hold in general, without having to assume the uniqueness of additively generated trading strategies nor wealth process. In addition, to impose further that the wealth process additively generated by the same regular function be unique, we invoke Proposition 3.4 from \cite{karatzas2016trading} and observe that if there exists a local martingale deflator for the market weight process $\pi(\cdot)$, then the process $\Gamma^G(\cdot)$ does not depend on the choice of $DG$ and is uniquely determined up to indistinguishability. 

\end{remark}

\subsection{Convex order and the gamma process}
In this section, we present a conjecture on the link between convex order and the gamma process $\Gamma(\cdot)$ in the theory of Functionally Generated Portfolio. 

\begin{conjecture}
    
Consider a generic $d-$dimensional semimartingale of market weights $\pi(\cdot)$ with continuous paths. Let $G_1,G_2 : \Delta^d \rightarrow [0,\infty)$ be a continuous regular function in $\mathcal{C}^2$. For all $t \geq 0$ and for $k = 1,2$, recall that the (continuous, adapted) gamma process functionally generated by $G_1$ and $G_2$ is defined as
$$\Gamma^{G_k}(t) : = G_k(X(0)) - G_k(X(t)) + \int_0^t \sum_{i = 1}^d v_{k_i}(t)dX_i(t)$$
for $k=1,2$, where 
$$v_{k_i}(t) := D_iG_k(X(t)), \ \  i = 1,...,d,t \geq 0$$
is integrable with respect to the $\Delta^d$-valued semimartingale $\pi(\cdot)$. 
\newline
\newline 
Now let $\psi,\phi$ be trading strategies \textbf{additively} generated by functions $G_1$, $G_2$. Fix a $t\geq 0$. If (continuous, adapted) process
$$\kappa (\cdot) := \Gamma^{G_1}(\cdot) - \Gamma^{G_2}(\cdot)$$
satisfies some form of $concavity^*$, then the probability measures $\mu,\nu$ defined by the investment proportion processes, i.e.
$$\mu_i(\cdot) := \frac{S_i(\cdot)\psi_i(\cdot)}{V^{\psi}(\cdot;S)}$$
$$\nu_i(\cdot) := \frac{S_i(\cdot)\phi_i(\cdot)}{V^{\phi}(\cdot;S)}$$
for $i = 1,2,...,d$ have finite second moments and are in a weaker form of convex order, which we call ``quasi-convex order" or the ``$\Gamma$-order".
\end{conjecture}

Below we give a motivation for Conjecture 6.3.1.

Let $\psi,\phi$ be trading strategies additively generated by $G_1,G_2$. Then by Proposition 2.2 and Proposition 4.3 in \cite{karatzas2016trading},
$$\mu_i(t) := \frac{S_i(t)\psi_i(t)}{V^{\psi}(t;S)} =\frac{S_i(t)\psi_i(t)}{\Sigma(t)V^{\psi}(t;\pi)} =  \frac{S_i(t)\psi_i(t)}{\Sigma(t)(G_1(\pi(t))+\Gamma^{G_1}(t))}$$
$$\nu_i (t) := \frac{S_i(t)\phi_i(t)}{V^{\phi}(t;S)} =\frac{S_i(t)\phi_i(t)}{\Sigma(t)V^{\phi}(t;\pi)} =  \frac{S_i(t)\phi_i(t)}{\Sigma(t)(G_2(\pi(t))+\Gamma^{G_2}(t))}$$
Again, by Proposition 4.3 of Karatzas and Ruf [2017] \cite{karatzas2016trading} we can write
\begin{equation}
\begin{split}
\psi_i(t) & = V^\psi(t;\pi)+D_iG_1(\pi(t))-\sum_{j=1}^d\pi_j(t)D_jG_1(\pi(t))\\
 & = G_1(\pi(t))+\Gamma^{G_1}(t)+D_iG_1(\pi(t))-\sum_{j=1}^d\pi_j(t)D_jG_1(\pi(t))
\end{split}
\end{equation}
and
\begin{equation}
\begin{split}
\phi_i(t) & = V^\phi(t;\pi)+D_iG_2(\pi(t))-\sum_{j=1}^d\pi_j(t)D_jG_2(\pi(t))\\
 & = G_2(\pi(t))+\Gamma^{G_2}(t)+D_iG_2(\pi(t))-\sum_{j=1}^d\pi_j(t)D_jG_2(\pi(t))
\end{split}
\end{equation}
Then we write
\begin{equation}
\begin{split}
\mu_i(t) & = \pi_i(t) \frac{V^\psi(t;\pi)+D_iG_1(\pi(t))-\sum_{j=1}^d\pi_j(t)D_jG_1(\pi(t))}{V^\psi(t;\pi)} \\
& = \pi_i(t) (1+ \frac{D_iG_1(\pi(t))-\sum_{j=1}^d\pi_j(t)D_jG_1(\pi(t))}{G_1(\pi(t))+\Gamma^{G_1}(t)})
\end{split}
\end{equation}
and so similarly,
\begin{equation}
\begin{split}
\nu_i(t) = \pi_i(t) (1+ \frac{D_iG_2(\pi(t))-\sum_{j=1}^d\pi_j(t)D_jG_2(\pi(t))}{G_2(\pi(t))+\Gamma^{G_2}(t)})
\end{split}
\end{equation}
\newline 
\newline 
since $$\frac{S_i(t)}{\Sigma(t)} = \pi_i(t).$$ 

In the theorem, we suppose process $\kappa(\cdot) = \Gamma^{G_1}(\cdot) -\Gamma^{G_2}(\cdot)$ to follow certain form of concavity. It is not yet clear what this concavity may be. Some ideas for the concavity include that it be defined as a logarithm transform of a random process $X$ that is stochastically concave, i.e. 
$$\E[wX(t_1) + (1-w)X(t_2)] \leq w\E[X(t_1)] + (1-w)\E[X(t_2)]$$ for all $t_1,t_2 \geq 0$ and $w \in [0,1]$. Intuitively, we want a form of bounded concavity in $\kappa(\cdot)$ because we want to existence of a dominating constant $D \in \mathbb{R}$ or another bounded function $y(t)$ such that for all values in $\text{supp}(\kappa(t))$, $\kappa(t) \leq D$ or $\kappa(t) \leq y(t)$. Since the $\Gamma^G$-process is interpreted as the quadratic variation term of its generating function $G$, it is a measure of spread. By saying that the difference process $(\kappa_t)_{t\geq 0}$ between two Gamma processes is bounded above by some constant (or some bounded function) means that the spread of the first portfolio generated by $G_1$ is no bigger than the second one generated by $G_2$ to a certain extent. This is weaker than the strict definition of convex order, since for $\Gamma$-order, it allows a certain amount of ``leeway" up to the dominating constant $D$ (or dominating bounded function $\gamma(t)$). In this sense, we call the form of ``relaxed" convex order the $\Gamma-$order a quasi convex order.

\section{Conclusion}
We have introduced how an arbitrage strategy can be determined through leveraging the absence of convex order in the market. We developed model-independent arbitrage strategy in trading with European options via optimal transport and also offered an overview of relative arbitrage between Functionally Generated Portfolios in Stochastic Portfolio Theory as well as how this could be related to a form of convex order in the market. Future works include formally proving or disproving Conjecture 6.3.1, back-testing model-independent trading strategies in Section 5 with real-world data, as well as developing more efficient numerical methods for the calculation of optimal $\rho$ as well as recovering convex function $\hat{f}$ in multi-dimension. 

\section{Appendix}
\subsection{A. Derivation of the TPE Acquisition Function}
We derive here the EI-based acquisition function for TPE. That is, 
\begin{align*}
EI_y^{*} \propto (\gamma + \frac{g(x)}{l(x)}(1-\gamma))^{-1}.
\end{align*}

\begin{proof}
Recall that the TPE algorithm parametrizes $p(x,y)$ as $p(y)p(x|y).$ Using Bayes' rule and by definition of the EI criterion, we have that $$EI_y^{*}(x) =\int_{-\infty}^{y^*}(y^*-y) \frac{p(x|y)p(y)}{p(x)}dy.$$ Let $\gamma$ be the cut-off quantile. Recall that $p(y < y^*) = \gamma$ and that TPE defines $p(x|y)$ using two densities over the configuration space $\mathcal{X}$: 
$$p(x|y)= 
\begin{cases}
    l(x),& \text{if } y < y^*\\
    g(x),              & \text{otherwise}
\end{cases}$$
Then we can write $$p(x) = \int_{\mathbb{R}}p(x|y)p(y)dy = \gamma y^* l(x)+(1-\gamma)g(x).$$ Therefore, $$\int_{-\infty}^{y^*}(y^*-y)p(x|y)p(y)dy = l(x) \int_{-\infty}^{y^*}(y^*-y)p(y)dy = \gamma y^*l(x)-l(x)\int_{-\infty}^{y^*}p(y)dy.$$ It follows that 
\begin{equation*}
\begin{split}
EI_y^{*} &= \int_{-\infty}^{y^*} (y^*-y) \frac{p(x|y)p(y)}{p(x)}dy \\ &= \int_{-\infty}^{y^*} (y^*-y) \frac{p(x|y)p(y)}{\gamma y^* l(x)+(1-\gamma)g(x)}dy \\ &= \frac{\gamma y^*l(x)-l(x)\int_{-\infty}^{y^*}p(y)dy}{\gamma l(x)+(1-\gamma)g(x)} \\ &\propto (\gamma + \frac{g(x)}{l(x)}(1-\gamma))^{-1}.
\end{split}
\end{equation*}

\end{proof}

\subsection{B. Derivation of the Poisson Equation with Neumann Boundary Conditions}
Assume that the domain $\Omega \subset \mathbb{R}^n$ is bounded and connected. Assume that $\hat{f}: \mathbb{R}^2 \rightarrow \mathbb{R}$ is sufficiently smooth. Given an empirically observed gradient field $g$ in 2-dimensional, we consider the general case where the values of $\hat{f}$ on the boundary are not known. Let $\phi: \mathbb{R}^2 \rightarrow \mathbb{R}$ be an arbitrary continuously differentiable function representing the random noise. We provide a proof for the equivalence of inverse gradient calculation in 2D with the Poisson equation with Neumann boundary conditions. 

\begin{proof}
Consider
\begin{equation*}
\begin{split}
\|\nabla (\hat{f}+ \phi)-g )\|^2 &= \|\nabla \hat{f}+ \nabla \phi (\textbf{x})-g(\textbf{x})\|^2 \\ &= \int_\Omega |\nabla \hat{f}-g(x)|^2d\textbf{x} + 2 \int_{\Omega}(\nabla \hat{f}(\textbf{x})-g(\textbf{x}))\cdot \nabla \phi(\textbf{x})d\textbf{x} \\ &+ \int_\Omega |\nabla \phi(\textbf{x})|^2d\textbf{x}.
\end{split}
\end{equation*}
Since $\arg\min_{\hat{f}}\|\nabla \hat{f}-g\|$ must also minimize $\arg\min_{\hat{f}}\|\nabla (\hat{f}+\phi)-g\|$ (consider taking a constant function $\phi(\textbf{x}) = 0$), this means that the following condition must be satisfied for the optimal $\hat{f}$:
\begin{equation*}
\begin{split}
\int_{\Omega}(\nabla \hat{f}(\textbf{x})-g(\textbf{x}))\cdot \nabla \phi(\textbf{x})d\textbf{x} = 0
\end{split}
\end{equation*}
Decomposing a line integral from the LHS, we have
\begin{equation*}
\begin{split}
0 &= \int_{\Omega}(\nabla \hat{f}(\textbf{x})-g(\textbf{x}))\cdot \nabla \phi(\textbf{x})d\textbf{x} \\ &= \int_{\partial \Omega} (\frac{\partial f}{\partial n}(\textbf{x})-\textbf{n}\cdot g(\textbf{x})) \phi(\textbf{x}) dS \\ &- \int_\Omega (\nabla^2f(\textbf{x})-\textbf{n} \cdot \nabla g(\textbf{x}))\phi(\textbf{x})d\textbf{x}.
\end{split}
\end{equation*}
Since we want the above equality to hold for every function $\phi: \mathbb{R}^2 \rightarrow \mathbb{R}$, we recovered the Neumann boundary condition:
$$\begin{cases}
    \nabla^2\hat{f}(\textbf{x}) = \nabla \cdot g(\textbf{x}),& \forall \textbf{x} \in \Omega\\
    \nabla \hat{f}(\textbf{x}) \cdot \textbf{n} = \textbf{n} \cdot g(\textbf{x}), & \forall\textbf{x} \in \partial \Omega,
\end{cases}$$
\end{proof}

\bibliographystyle{plainnat} 
\bibliography{main}          

\end{document}